\documentclass[12pt]{article}
\usepackage[T1]{fontenc}
\usepackage{pgf}
\usepackage{tikz}
\usetikzlibrary{arrows,automata,positioning,calc}
\tikzset{
 mainNode/.style =
    { circle
    , draw
    }
}
\usepackage[latin1]{inputenc}

\usepackage{float}

\usepackage{mathrsfs}
\usepackage{lmodern}
\usepackage{paralist}
\usepackage{url}

\usepackage{amsmath,amssymb,amsthm} 

\usepackage{graphicx}
\usepackage[left=3cm,right=3cm,top=4cm,bottom=4cm]{geometry}

\usepackage{booktabs} 

\usepackage[small]{caption}
\usepackage{subcaption}

\usepackage{enumerate}
\usepackage{xspace}

\usepackage{microtype} 

\newcommand{\mat}[1]{\ensuremath{\boldsymbol{#1}}}
\newcommand{\vect}[1]{\ensuremath{\boldsymbol{#1}}}
\newcommand{\bvect}{\ensuremath{\boldsymbol{B}}}
\newcommand{\trans}[1]{\ensuremath{{#1}^{\scriptscriptstyle \mathsf{T}}}}
\newcommand{\norm}[1]{\ensuremath{\left \| #1 \right \|}}

\newcommand{\inv}[1]{\ensuremath{{#1}^{-1}}}
\newcommand{\etal}{\emph{et al.\!}}
\renewcommand{\th}{\ensuremath{^{\textrm{th}}}}
\renewcommand{\epsilon}{\varepsilon}
\newcommand{\NMI}{\ensuremath{\mathrm{NMI}}}
\def\CPP{{C\nolinebreak[4]\hspace{-.05em}\raisebox{.4ex}{\tiny\bf ++}}}

\newtheorem{theorem}{Theorem}

\newtheorem{proposition}{Proposition}

\theoremstyle{definition}

\usepackage{authblk}
\usepackage{lmodern}

\title{Overlapping Communities in Social Networks\thanks{Research 
funded by DFG Project RO 927/13-1 ``Pragmatic Parameterized Algorithms.''}}
\author{Jan Dreier} 
\author{Philipp Kuinke} 
\author{Rafael Przybylski} 
\author{Felix Reidl} 
\author{Peter Rossmanith} 
\author{Somnath Sikdar}
\affil{\small \emph{Theoretical Computer Science} \\ \emph{RWTH Aachen University, 52074 Aachen, Germany}}
\date{}

\begin{document}
\maketitle

\begin{abstract}
Complex networks can be typically broken down into groups or modules. Discovering 
this ``community structure'' is an important step in studying the large-scale 
structure of networks. Many algorithms have been proposed for community detection 
and benchmarks have been created to evaluate their performance. Typically algorithms 
for community detection either partition the graph (non-overlapping 
communities) or find node covers (overlapping communities). 

In this paper, we propose a particularly simple semi-supervised learning 
algorithm for finding out communities. In essence, given the community information of a small 
number of ``seed nodes'', the method uses random walks from the seed nodes 
to uncover the community information of the whole network. The algorithm runs 
in time $O(k \cdot m \cdot \log n)$, where $m$ is the number of edges; $n$ 
the number of links; and $k$ the number of communities in the network. 
In sparse networks with $m = O(n)$ and a constant number of communities, this 
running time is almost linear in the size of the network. Another important 
feature of our algorithm is that it can be used for either non-overlapping 
or overlapping communities. 

We test our algorithm using the LFR benchmark created by Lancichinetti, 
Fortunato, and Radicchi~\cite{LFR08} 
specifically for the purpose of evaluating such algorithms. Our algorithm 
can compete with the best of algorithms for both non-overlapping 
and overlapping communities as found in the comprehensive study of 
Lancichinetti and Fortunato~\cite{LF09}.

\end{abstract}

\section{Introduction}
Many real-world graphs that model complex systems exhibit an organization 
into subgraphs, or \textit{communities} that are more densely connected on the inside than between each other. 
Social networks such as Facebook and LinkedIn divide into groups of friends 
or coworkers, or business partners; scientific collaboration networks divide 
themselves into research communities; the World Wide Web divides into groups 
of related webpages. The nature and number of communities provide 
a useful insight into the structure and organization of networks. 

Discovering the community structure of networks is an 
important problem in network science and is the subject 
of intensive research~\cite{GN02, GN04, CNM04, RCC04, DM04, PDFV05, NL07, 
BGLL08, RB08, RN09}. Existing community detection algorithms are 
distinguished by whether they find partitions of the node set 
(non-overlapping communities) or node covers (overlapping communities). 
Typically finding overlapping communities is a much harder problem and most of the 
earlier community detection algorithms focused on finding disjoint 
communities. A comparative analysis of several community detection algorithms 
(both non-overlapping and overlapping) was presented by Lancichinetti and Fortunato 
in~\cite{LF09}. In this paper we closely follow their test framework, 
also called the LFR-benchmark.

The notion of a community is a loose one and currently there is no 
well-accepted definition of this concept. A typical approach is to define an 
objective function on the partitions of the node set of the network 
in terms of two sets of edge densities: the density of the 
edges within a partite set (intra-community edges) and the density of edges across partitions 
(inter-community edges). The ``correct'' partition is the one that maximizes this 
function. Various community detection algorithms formalize this
informal idea differently. One of the very first algorithms by
Girvan and Newman~\cite{GN02} introduced a measure known as \textit{modularity}
which, given a partition of the nodes of the network, compares the fraction of 
inter-community edges with the edges that would be present had they been 
rewired randomly preserving the node degrees. Other authors such as Palla 
\etal~\cite{PDFV05} declare communities as node sets that formed 
by overlapping maximal cliques. Rosvall and Bergstrom~\cite{RB08} 
define the goodness of a partition in terms of the number of bits required to 
describe per step of an infinite random walk in the network, the intuition being 
that in a ``correct'' partition, a random walker is likely to spend more time 
within communities rather than between communities, thereby decreasing the 
description of the walk.  

A severe restriction of many existing community detection algorithms 
is that they are too slow. Algorithms that optimize modularity typically 
take $O(n^2)$, even on sparse networks. The overlapping clique finding 
algorithm of Palla \etal~\cite{PDFV05} take exponential time in the worst case.
In other cases, derivation of worst-case running time bounds are ignored. 

\paragraph{Our contribution.}
Given that it is unlikely that users of community detection algorithms 
would unanimously settle on one definition of what constitutes a community, 
we feel that existing approaches ignore the \emph{user perspective}.
To this end, we chose to design an algorithm that takes the network structure 
as well as user preferences into account. 
The user is expected to classify a small set of nodes of the network 
into communities (which may be 6--8\% of the nodes of each community).
Obviously this is possible only when the user has some 
information about the network, such as its semantics, which nodes 
are important and into which communities they are classified. 

Such situations are actually quite common. The user might have data only 
on the leading scientific authors in a co-authorship network 
and would like to find out the research areas of the remaining members of the network. 
He may either be interested in a broad partition of the network into into its main fields
or a fine grained decomposition into various subfields.
By labeling the known authors accordingly,
the user can specify which kind of partition he is interested in.
Another example would be the detection of trends in a social network.
Consider the case where one knows the political affiliations of some people 
and aims to discover political spectrum of the whole network, 
for example, to predict the outcome of an election.

Another scenario where this may be applicable is in recommendation systems. 
One might know the preferences of some of the users of an online retail 
merchant possibly because they purchase items much more frequently than others. 
One could then use this in the network whose nodes consist of users, with two 
nodes connected by an edge if they represent users that had purchased similar products in the past. 
The idea now would be to use the knowledge of the preferences of a few to 
predict the preferences of everyone in the network.

An important characteristic of algorithms surveyed in~\cite{LF09} 
is that the algorithms either find disjoint communities or overlapping 
ones. Most algorithms solve the easier problem of finding disjoint communities. 
The ones that are designed to find overlapping communities such as the overlapping clique finding 
algorithm of Palla \etal~\cite{PDFV05} do not seem to yield very good results (see~\cite{LF09}).
Our algorithm naturally extends to the overlapping case. Of course, there is a higher 
price that has to be paid in that the number of nodes that need to be classified by the user 
typically is larger (5\% to 10\% of the nodes per community). The algorithm, however,
does not need any major changes and we view this is as an aesthetically pleasing 
feature of our approach. 

Thirdly, in many other approaches, the worst-case running time of the algorithms 
is neither stated nor analyzed. We show that our algorithm runs in time $O(k \cdot m \cdot \log n)$, 
where $k$ is the number of communities to be discovered (which is supplied by the user), 
$n$ and $m$ are the number of nodes and edges in the network. In the case of sparse graphs 
and a constant number of communities, the running time is $O(n \cdot \log n)$. Given that even 
an $O(n^2)$ time algorithm is too computationally expensive on many real world graphs, 
a nearly linear time algorithm often is the only feasible solution.   

Finally, we provide an extensive experimental evaluation of our algorithm on the 
LFR benchmark. In order to ensure a fair comparison with other algorithms 
reviewed in~\cite{LF09}, we choose all parameters of the benchmark as in the original paper.

This paper is organized as follows. In Section~\ref{sec:major_algorithms}, we review some 
of the more influential algorithms in community detection. In Section~\ref{sec:algorithm}, we 
describe our algorithm and analyze its running time. In Sections~\ref{sec:experiment_setup} 
and~\ref{sec:experiment_results}, we present our experimental results. Finally we conclude 
in Section~\ref{sec:conclusions} with possibilities of how our approach might be extended.

\section{The Major Algorithms} \label{sec:major_algorithms}
In what follows, we briefly describe some common algorithms 
for community detection. We are particularly interested in the performance 
of these algorithms as reported in the study by Lancichinetti and 
Fortunato~\cite{LF09} on their LFR benchmark graphs. 

\paragraph{The Girvan-Newman algorithm.} 
One of the very first algorithms for detecting disjoint communities 
was invented by Girvan and Newman~\cite{GN02, GN04}. Their 
algorithm takes a network and iteratively removes edges based 
on a metric called \emph{edge betweenness}. The betweenness of an 
edge is defined as the number of shortest paths between vertex pairs 
that pass through that edge. After an edge is removed, betweenness 
scores are recalculated and an edge with maximal score is deleted. 
This procedure ends when the \textit{modularity} of the resulting partition
reaches a maximum. Modularity is a measure that estimates the quality 
of a partition by comparing the network with a so-called ``null model''
in which edges are rewired at random between the nodes of the network 
while each node keeps its original degree.

Formally, the \emph{modularity} of a partition is defined as:
\begin{equation}\label{eqn:modularity}
	Q = \frac{1}{2m} \sum_{i, j} \left ( A_{i j} - \frac{d_i d_j}{2m} \right ) \delta(i, j),
\end{equation}
where $A_{ij}$ represent the entries of the adjacency matrix of the network; $d_i$ is the 
degree of node $i$; $m$ is the number of edges in the network; and $\delta(i, j) = 1$ if nodes
$i$ and $j$ belong to the same set of the partition and $0$ otherwise. The term $d_i d_j / 2m$ 
represents the expected number of edges between nodes $i$ and $j$ if we consider a random model
in which each node $i$ has $d_i$ ``stubs'' and we are allowed to connect stubs at random to form edges. 
This is the null model against which the within-community edges of the partition is compared against.
The worst-case complexity of the Newman-Girvan algorithm is dominated by the time taken 
to compute the betweenness scores and is $O(m n)$ for general graphs and $O(n^2)$ for sparse 
graphs~\cite{Bra01}.

\paragraph{The greedy algorithm for modularity optimization by Clauset \etal~\cite{CNM04}.}
This algorithm starts with each node being the sole member of a community of one, and 
repeatedly joins two communities whose amalgamation produces the largest increase in modularity. 
The algorithm makes use of efficient data structures and has a running time of $O(m \log^2 n)$, 
which for sparse graphs works out to $O(n \log^2 n)$. 

\paragraph{Fast Modularity Optimization by Blondel \etal} 
The algorithm of Blondel \etal~\cite{BGLL08} consists of two phases which are repeated iteratively. 
It starts out by placing each node in its own community and then locally optimizing the modularity 
in the neighborhood of each node. In the second phase, a new network is built whose nodes are the 
communities found out in the first phase. That is, communities are replaced by ``super-nodes''; the 
within-community edges are modeled by a (weighted) self-loop to the super-node; and the between-community 
edges are modeled by a single edge between the corresponding super-nodes, with the weight being 
the sum of the weights of the edges between these two communities. 
Once the second phase is complete, the entire procedure is repeated until the modularity does not 
increase any further. The algorithm is efficient due to the fact that one can quickly compute the 
change in modularity obtained by moving an isolated node into a community.

Lancichinetti and Fortunato opine that modularity-based methods in general have a rather poor 
performance, which worsens for larger networks. The algorithm due to Blondel \etal\ performs well 
probably due to the fact that the estimated modularity is not a good approximation of the real one~\cite{LF09}.  


\paragraph{The CFinder algorithm of Palla \etal} One of the first algorithms that dealt with overlapping 
communities was proposed by Palla \etal~\cite{PDFV05}. They define a community to be a set of nodes 
that are the union of $k$-cliques such that any one clique can be reached from another via a series of 
adjacent $k$-cliques. Two $k$-cliques are defined to be adjacent if they share $k - 1$ nodes. 

The algorithm first finds out all maximal cliques in the graph, which takes exponential-time
in the worst case. It then creates a symmetric clique-clique overlap matrix $\mat{C}$ which is a square matrix whose 
rows and columns are indexed by the set of maximal cliques in the graph and whose $(i, j)^{\mathrm{th}}$ entry
is the number of vertices that are in both the $i^{\mathrm{th}}$ and $j^{\mathrm{th}}$ clique. This matrix is then modified into 
a binary matrix by replacing all those entries with value less than $k - 1$ by a $0$ and the remaining entries by 
a $1$. The final step is to find the connected components of the graph represented by this binary symmetric matrix
which the algorithm reports as the communities of the network. 

The authors report to have tested the algorithm on various networks including the protein-protein interaction 
network of \emph{Saccharomyces cerevisiae}\footnote{A species of yeast used in wine-making, baking, and brewing.} 
with $k = 4$; the co-authorship network of the Los Alamos condensed matter archive (with $k = 6$). 
Lancichinetti and Fortunato report that CFinder did not perform particularly well on the LFR benchmark and 
that its performance is sensitive to the sizes of community (but not the network size). For networks will 
small communities it has a decent performance, but has a worse performance on those with larger communities. 

\paragraph{Using random walks to model information flow.} Rosvall and Bergman~\cite{RB08} approach the 
problem of finding communities from an information-theoretic angle. They transform the problem 
into one of finding an optimal description of an infinite random walk in the network.  Given a fixed 
partition $M$ of $n$ nodes into $k$ clusters, Rosvall and Bergman use a two-level code where each 
cluster is assigned a unique codeword and each node is assigned a codeword which is unique per community.
One can now define the average number of bits per step that it 
takes to describe an infinite random walk on the network partitioned according to $M$. The intuition 
is that a random walker is statistically likely to spend more time within clusters than between clusters
and therefore the ``best'' partition corresponds to the one which has the shortest possible description.
An approximation of the best partition is found out using a combination of a greedy search heuristic followed by simulated 
annealing. Lancichinetti and Fortunato report that this algorithm (dubbed Infomap) was the best-performing 
among all other community detection algorithms on their benchmark.  

\paragraph{}

\section{The Algorithm} \label{sec:algorithm}
We assume that the complex networks that we deal with are modeled as connected, undirected graphs.  
The algorithm receives as input a network and a set of nodes such that there is at least one node 
from each community that we are aiming to discover. These nodes are called \emph{seed nodes} and 
it is possible that a particular seed node belongs to multiple communities.

The affinity of a node in the network to a community is 1 if it belongs to it; 
if it does not belong to it, it has an affinity of 0. We allow intermediate affinity 
values and view these as specifying a \emph{partial belonging}. The user specifies the affinities 
of the seed nodes for each of the communities. For all other nodes, called 
\emph{non-seed nodes}, we want deduce the affinity to each community using 
the information given by the seed nodes' affinities and the network structure. 
The main idea is that non-seed nodes should adopt the affinities of seed nodes 
within their close proximity. We define a proximity measure based on random walks: 
Each random walk starts at a non-seed node, traverses through the graph, and ends 
as soon as it reaches a seed node. The affinity of a non-seed node~$u$ for a given 
community is then the weighted sum of the affinities of the seed nodes for 
that community and reachable by a random walk starting at $u$, the weights 
being the probabilities that a random walk from $u$ ends up at a certain seed node.

Each step of a random walk can be represented as the iterated product of 
a transition matrix $\mat{P}$. The result of the (infinite) walk itself 
can be expressed as $\lim_{k \rightarrow \infty}{\mat{P}^k}$.
One of the contributions of this paper is to show how the calculation
of these limits can be reduced to solving a symmetric, diagonally dominant
system of linear equations (with different right-hand-sides per community),
which can be done in $O(m \log n)$ time, where~$m$ is the number of edges 
in the graph. The fact that such systems can be solved in almost linear time 
was discovered by Spielman and 
Teng~\cite{ST04,EEST05,ST08,KMP10,KMP11,Vis13}. If we assume that our networks are 
sparse in the sense that $m = O(n)$, the running time of our algorithm can  
be bounded by $O(n \log n)$.

\subsection{Absorbing Markov Chains and Random Walks}
We now provide a formal description of our model. 
The input is an undirected, connected graph $G = (V,E)$ with nodes $v_1, \ldots, v_n$, 
$m$ edges and a nonempty set of $s$ seed nodes. We also know that there are~$k$ 
(possibly overlapping) communities which we want to discover.

The community information of a node~$v$ is represented 
by a $1 \times k$ vector called the \emph{affinity vector} of~$v$, denoted 
by 
\[
	\bvect (v) = \trans{\left ( \alpha (v, 1), \ldots, \alpha (v, k) \right )}.
\] 
The entry $\alpha (v, l)$ of the affinity vector represents the affinity of node~$v$ to community~$l$.
It may be interpreted as the probability that a node belongs to this community.
We point out that $\sum_{i = 1}^k \alpha (v, l)$ need not be~$1$.
An example of this situation is when~$v$ belongs to multiple 
communities with probability~$1$. 
The user-chosen affinity vectors of all seed nodes are part of the input.
The objective is to derive the affinity vectors of all non-seed nodes. 

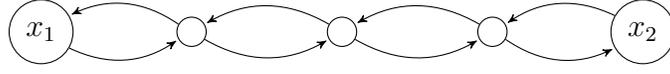
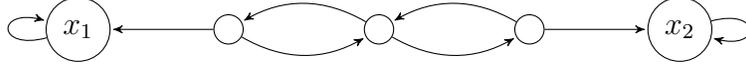
\begin{figure}
\centering
\begin{subfigure}{\textwidth}
    \centering
    \begin{tikzpicture}
        [ auto
        , ->
        , >=stealth'
        , shorten >=1pt
        , node distance=2cm
        ]
        \node[mainNode] (1) {$x_1$};
        \node[mainNode] (2) [right of=1] {};
        \node[mainNode] (3) [right of=2] {};
        \node[mainNode] (4) [right of=3] {};
        \node[mainNode] (5) [right of=4] {$x_2$};

        \path[every node/.style={font=\sffamily\small}]
        (1) edge [bend right] node {} (2)
        (2) edge [bend right] node {} (3)
        (3) edge [bend right] node {} (4)
        (4) edge [bend right] node {} (5)

        (5) edge [bend right] node {} (4)
        (4) edge [bend right] node {} (3)
        (3) edge [bend right] node {} (2)
        (2) edge [bend right] node {} (1)
        ;
    \end{tikzpicture}
    \caption{Example graph with seed nodes $x_1$, $x_2$.}
\end{subfigure}
\\[0.5cm]
\begin{subfigure}{\textwidth}
    \centering
    \begin{tikzpicture}
        [ 
          auto
        , ->
        , >=stealth'
        , shorten >=1pt
        , node distance=2cm
        ]
        \node[mainNode] (1) {$x_1$};
        \node[mainNode] (2) [right of=1] {};
        \node[mainNode] (3) [right of=2] {};
        \node[mainNode] (4) [right of=3] {};
        \node[mainNode] (5) [right of=4] {$x_2$};

        \path[every node/.style={font=\sffamily\small}]
        (1) edge [loop left]  node {} (1)
        (5) edge [loop right] node {} (5)

        (2) edge [bend right] node {} (3)
        (3) edge [bend right] node {} (4)
        (4) edge node {} (5)

        (4) edge [bend right] node {} (3)
        (3) edge [bend right] node {} (2)
        (2) edge node {} (1)
        ;
    \end{tikzpicture}
    \caption{Example graph with seed nodes $x_1$, $x_2$ after transformation.}
\end{subfigure}
\caption[Removal of outgoing edges of seed nodes in a graph]
{Remove outgoing edges and add self-loop for all seed nodes in an example graph. A random walk reaching $x_1$ or $x_2$ will stay there forever.}
\label{fig:modgraph}.
\end{figure}

Since we require the random walks to end as soon as they reach a seed node, we transform 
the undirected graph~$G$ into a directed graph $G'$ as follows: replace each undirected 
edge $\{u, v\}$ by arcs $(u, v)$ and $(v, u)$; then for each seed node, remove its 
outarcs and add a self-loop. This procedure is illustrated in Figure~\ref{fig:modgraph}.

Random walks in this graph can be modelled by an $n \times n$ transition matrix $\mat{P}$, with
\begin{equation}\label{eqn:defining_prob}
	\mat{P}(i, j) = \left \{ 
							\begin{array}{ll}
                                \frac{1}{\deg_{G'} (v_i)} & \mbox{ if } (v_i, v_j) \in E(G') \\
								0			& \mbox{ otherwise,}
							\end{array}
					\right.
\end{equation}
where $\deg_{G'} (v)$ is the degree of node $v$ in the directed graph~$G'$.
The entry $\mat{P}(i, j)$ represents the transition-probability from node $v_i$ to $v_j$.
Additionally, $\mat{P}^r (i, j)$ may be interpreted 
as the probability that a random walk starting at node~$v_i$ will end up at node~$v_j$ 
after~$r$ steps. 

Assume that the nodes of~$G'$ are labeled $u_1, \ldots, u_{n - s}, x_{1}, \ldots, x_{s}$, 
where $u_1, \ldots, u_{n - s}$ are the non-seed nodes and $x_1, \ldots, x_s$ are the seed nodes. 
We can now write the transition matrix $\mat{P}$ in the following canonical form:
\begin{equation}\label{eqn:canonical_form_P}
	\mat{P} = 	\left [ \begin{array}{ll}
						\mat{Q}  & \mat{R} \\
						 \mat{0}_{s \times (n - s)} & \mat{I}_{s \times s}
						\end{array}
				\right ],
\end{equation}
where~$\mat{Q}$ is the $(n - s) \times (n - s)$ sub-matrix that represents the transition 
from non-seed nodes to non-seed nodes; $\mat{R}$ is the $(n - s) \times s$ sub-matrix 
that represents the transition from non-seed nodes to seed nodes. The $s \times s$ identity 
matrix $\mat{I}$ represents the fact that once a seed node is reached, one cannot transition away 
from it. Here $\mat{0}_{s \times (n - s)}$ represents an $s \times (n - s)$ matrix of zeros. 
Since each row of $\mat{P}$ sums up to~$1$ and all entries are positive, this matrix is stochastic.

It is well-known that such a stochastic matrix represents what is known as an 
\emph{absorbing Markov chain} (see, for example, Chapter~11 of Grinstead and Snell~\cite{GS98}).  
A Markov chain is called absorbing if it satisfies two conditions: 
It must have at least one absorbing state~$i$, where state~$i$ is defined 
to be absorbing if and only if $\mat{P}(i,i) = 1$ and $\mat{P}(i, j) = 0$ 
for all $j \neq i$. Secondly, it must be possible to transition from every state to 
some absorbing state in a finite number of steps.
It follows directly from the construction of the graph $G'$ and the fact that the 
original graph was connected, that random walks in $G'$ define an absorbing Markov chain. 
Here, the absorbing states correspond to the set of seed nodes.

For any non-negative $r$, one can easily show that:
\begin{equation}\label{exp:rth_product_of_P}
	\mat{P}^r = \left [ \begin{array}{ll}
						\mat{Q}^r  					& \sum_{i = 0}^{r - 1}\mat{Q}^i \cdot \mat{R} \\
						 \mat{0}_{s \times (n - s)} & \mat{I}_{s \times s}
						\end{array}
				\right ].
\end{equation}  
Since we are dealing with infinite random walks, we are interested in the following 
property of absorbing Markov chains.
\begin{proposition}\label{prop:limiting_Q}
	Let $\mat{P}$ be the $n \times n$ transition matrix that defines an absorbing Markov chain
	and suppose that $\mat{P}$ is in the canonical form specified by equation~(\ref{eqn:canonical_form_P}). 
	Then
    \begin{equation}
        \lim_{r \to \infty} \mat{P}^r = \left [ \begin{array}{ll}
            \mat{0}_{(n - s) \times (n - s)} & (\mat{I} - \mat{Q})^{-1} \cdot \mat{R} \\
            \mat{0}_{s \times (n - s)}       & \mat{I}_{s \times s}
        \end{array}
        \right ].
    \end{equation}  
\end{proposition}
Intuitively, every random walk starting at a non-seed node eventually 
reaches some seed node where it is ``absorbed.'' The probability 
that such an infinite random walk starting at non-seed node~$u_i$ ends 
up at the seed node~$x_j$ is entry $(i, j)$ of the 
submatrix $\mat{X} := (\mat{I} - \mat{Q})^{-1} \cdot \mat{R}$. 

Now we can finally define the affinity vectors of non-seed nodes.
The affinity of non-seed node $u_i$ to a community~$l$ is defined as:
\begin{equation}\label{eqn:belonging_vector}
    \alpha (u_i, l) = \sum_{j = 1}^{s} \mat{X} (i, j) \cdot \alpha (x_j, l).
\end{equation}
The computational complexity of calculating these affinity values  
depends on how efficiently we can calculate the entries of $\mat{X}$, 
i.e., solve $(\mat{I} - \mat{Q})^{-1}$. In the next subsection, we show 
how to reduce this problem to that of solving a system of linear equations 
of a special type which takes time $O(m \cdot \log n)$, where $m$
is the number of edges in $G$.

\subsection{Symmetric Diagonally Dominant Linear Systems}

An $n \times n$ matrix $\mat{A} = [a_{ij}]$ is \emph{diagonally dominant} if 
\[
	|a_{ii}| \geq \sum_{j \neq i} {|a_{ij}|} \mbox{ for all } i = 1, \ldots, n.
\] 
A matrix is \emph{symmetric diagonally dominant (SDD)} if, in addition to the above, 
it is symmetric. For more information about matrices and matrix computations, 
see the textbooks by Golub and Van Loan~\cite{GvL13} and Horn and Johnson~\cite{HJ13}. 

An example of a symmetric, diagonally dominant matrix is the graph Laplacian. 
Given an unweighted, undirected graph~$G$, the \emph{Laplacian} of $G$ 
is defined to be 
\[
\mat{L}_G = \mat{D}_G - \mat{A}_G,
\] 
where $\mat{A}_G$ is the adjacency matrix of the graph~$G$ and $\mat{D}_G$ 
is the diagonal matrix of vertex degrees. 


A symmetric, diagonally dominant (SDD) system of linear equations is a system of 
equations of the form:
\[
	\mat{A} \cdot \vect{x} = \vect{b},
\]
where $\mat{A}$ is an SDD matrix, $\vect{x} = \trans{(x_1, \ldots, x_n)}$ 
is a vector of unknowns, and $\vect{b} = \trans{(b_1, \ldots, b_n)}$ is a vector of constants. 
There is near-linear time algorithm for solving such a system of linear equations 
and this result is crucial to the analysis of the running time of our algorithm. 

The solution of $n \times n$ system of linear equations takes $O(n^3)$ time 
if one uses Gaussian elimination. Spielman and Teng made a seminal contribution in this direction and 
showed that SDD linear systems can be solved in nearly-linear 
time~\cite{ST04,EEST05,ST08}. Spielman and Teng's algorithm (the ST-solver)
iteratively produces a sequence of approximate solutions which converge to the 
actual solution of the system $\mat{A} \vect{x} = \vect{b}$. The performance 
of such an iterative system is measured in terms of the time taken to reduce 
an appropriately defined approximation error by a constant factor. The time 
complexity of the ST-solver was reported to be at least $O(m \log^{15} n)$~\cite{KMP11}.  
Koutis, Miller and Peng~\cite{KMP10,KMP11} developed a simpler and faster algorithm 
for finding $\epsilon$-approximate solutions to SDD systems in time 
$\tilde{O}(m \log n \log (1/\epsilon) )$, where the $\tilde{O}$ notation hides 
a factor that is at most $(\log \log n)^2$. A highly readable account 
of SDD systems is the monograph by Vishnoi~\cite{Vis13}. We summarize the 
main result that we use as a black-box.  
\begin{proposition} \label{prop:SDD_systems} {{\rm \cite{KMP11,Vis13}}}
	Given a system of linear equations $\mat{A} \vect{x} = \vect{b}$, where $\mat{A}$
	is an SDD matrix, there exists an algorithm to compute $\tilde{\vect{x}}$  
	such that:
		\[
			\norm{\tilde{\vect{x}} - \vect{x}}_{\mat{A}} \leq \epsilon \norm{\vect{x}}_{\mat{A}}, 
		\]
	where $\norm{\vect{y}}_{\mat{A}} := \sqrt{\trans{\vect{y}} \mat{A} \vect{y}}$. The algorithm runs in 
	time $\tilde{O}(m \cdot \log n \cdot \log (1 / \epsilon) )$ time, where $m$ is the number of non-zero 
	entries in $\mat{A}$. The $\tilde{O}$ notation hides a factor of at most $(\log \log n)^2$.
\end{proposition} 


We can use Proposition~\ref{prop:SDD_systems} to upper-bound the time taken to solve
the linear systems, which are needed to calculate the affinity vectors defined in (\ref{eqn:belonging_vector}).

\begin{theorem}\label{theorem:computing_NR}
Given a graph~$G$, let $\mat{P}$ be the $n \times n$ transition matrix 
defined by equation~(\ref{eqn:defining_prob}) in canonical form 
(see equation~(\ref{eqn:canonical_form_P})). Then, one can compute 
the affinity vectors of all non-seed nodes in time $O(m \cdot \log n)$ per community, 
where~$m$ is the number of edges in the graph~$G$.
\end{theorem}  
\begin{proof}
Recall that we ordered the nodes of $G$ as $u_1, \ldots, u_{n - s}, x_1, \ldots, x_s$, 
where $u_1, \ldots, u_{n - s}$ denote the non-seed nodes and $x_1, \ldots, x_s$ denote 
seed nodes. Define $G_1 := G[u_1, \ldots, u_{n - s}]$, the subgraph induced by the non-seed nodes 
of~$G$. Let $\mat{A}_1$ denote the adjacency matrix of the graph $G_1$; let 
$\mat{D}_1$ denote the $(n - s) \times (n - s)$ diagonal matrix satisfying 
$\mat{D}_1(u_i, u_i) = \deg_{G}(u_i)$ for all $1 \leq i \leq n - s$.  That is, the 
entries of $\mat{D}_1$ are not the degrees of the vertices in the induced subgraph~$G_1$ 
but in the graph~$G$. We can then express 
$\mat{I} - \mat{Q}$ as 
\begin{equation} \label{eqn:I-Q}
	\mat{I}  - \mat{Q} = \inv{\mat{D}_1} (\mat{D}_1 - \mat{A}_1).
\end{equation}
Note that $\mat{D}_1 - \mat{A}_1$ is a symmetric and diagonally dominant matrix. 
Let us suppose that $\mat{X}$ is an $(n - s) \times s$ matrix such that 
\[
	\mat{X} = (\mat{I} - \mat{Q})^{-1} \cdot \mat{R}.
\]

Fix a community~$l$. Then the affinities of the non-seed nodes 
for community~$l$ may be written as:
\begin{align} \label{eqn:affinity}
	\left ( \begin{array}{c}
		\alpha (u_1, l) \\
		\vdots			\\
		\alpha (u_{n - s}, l)
	\end{array}	\right ) & = \sum_{j = 1}^{s} \alpha (x_j, l) \cdot \mat{X}_j \nonumber \\ 
						 & = \sum_{j = 1}^{s} \alpha (x_j, l) \inv{(\mat{I} - \mat{Q})} \cdot \mat{R}_j \nonumber \\ 
						 & = \inv{(\mat{I} - \mat{Q})} \cdot \sum_{j = 1}^{s} \alpha (x_j, l) \cdot \mat{R}_j,
\end{align}
where $\mat{X}_j$ and $\mat{R}_j$ denote the $j\th$ columns of $\mat{X}$ and $\mat{R}$, respectively. 
Using equation~(\ref{eqn:I-Q}), we may rewrite equation~(\ref{eqn:affinity}) as:
\begin{align}
	\inv{\mat{D}_1} (\mat{D}_1 - \mat{A}_1) \cdot \left ( \begin{array}{c}
		\alpha (u_1, l) \\
		\vdots			\\
		\alpha (u_{n - s}, l)
	\end{array}	\right ) & = \sum_{j = 1}^{s} \alpha (x_j, l) \cdot \mat{R}_j.
\end{align}
Finally, multiplying by $\mat{D}_1$ on both sides, we obtain
\begin{equation}\label{eqn:final_affinity}
	(\mat{D}_1 - \mat{A}_1) \cdot \vect{\alpha}_l = \mat{D}_1 \cdot \sum_{j = 1}^{s} \alpha (x_j, l) \cdot \mat{R}_j,
\end{equation}
where we used $\vect{\alpha}_l$ to denote the vector $\trans{\left ( \alpha (u_1, l), \ldots, \alpha (u_{n-s}, l) \right )}$.

Note that computing $\sum_{j = 1}^{s} \alpha (x_j, l) \cdot \mat{R}_j$ takes time $O(\tilde{m})$, where $\tilde{m}$ 
denotes the number of non-zero entries\footnote{This is almost the same as the 
number~$m$ of edges in $G$, but not quite, since while constructing $\mat{P}$ from the graph $G$, 
we add self-loops on seed nodes and delete edges between adjacent seed nodes, if any. However what is true is 
that $\tilde{m} \leq m + s \leq m + n$.} in \mat{P}. 
Computing the product of $\mat{D}_1$ and $\sum_{j = 1}^{s} \alpha (x_j, l) \cdot \mat{R}_j$ 
takes time $O(\tilde{m})$ so that the right hand side of equation~(\ref{eqn:final_affinity}) can 
be computed in time $O(\tilde{m})$. We now have a symmetric diagonally dominant system of linear equations 
which by Proposition~\ref{prop:SDD_systems} can be solved in time $O(\tilde{m} \cdot \log n)$. Therefore,
the time taken to compute the affinity to a fixed community is $O(\tilde{m} \cdot \log n) = O(m \log n)$,
which is what was claimed. Since we assume our networks to be sparse, $m = O(n)$, and 
the time taken is $O(n \cdot \log n)$ per community.  
%
\end{proof}

\section{Experimental Setup} \label{sec:experiment_setup}
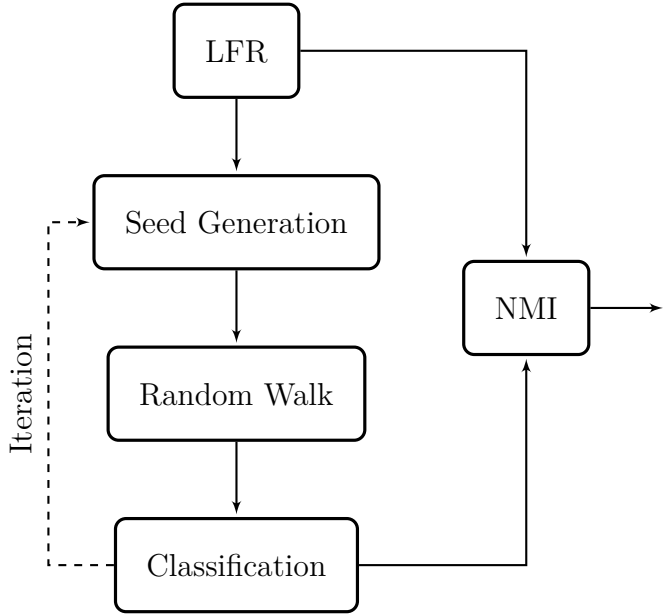
\begin{figure}
    \centering
    \begin{tikzpicture}[node distance=1cm, auto]  
        \tikzset{
            mynode/.style={
                rectangle,
                rounded corners,
                draw=black, 
                very thick, 
                inner sep=1em, 
                minimum size=3em, 
                text centered
            },
            myarrow/.style={
                ->, 
                >=latex', 
                shorten >=1pt, 
                thick
            },
            mylabel/.style={
                text width=7em,
                text centered
            },
            bendedarrow/.style={
                to path={(\tikztostart) -- ++(#1,0pt) \tikztonodes |- (\tikztotarget) },
                pos=0.5
            }
        }  
        

        \node[mynode]                  (lfr)            {LFR};  
        \node[mynode, below=of lfr]    (seeds)          {Seed Generation};  
        \node[mynode, below=of seeds]  (walk)           {Random Walk};  
        \node[mynode, below=of walk]   (classification) {Classification};  

        \coordinate (middle) at ($(seeds)!0.5!(walk)$);
        \node[mynode, right=of middle, xshift=2cm] (nmi) {NMI};  


        \draw[myarrow] (lfr)   -- (seeds);	
        \draw[myarrow] (seeds) -- (walk);	
        \draw[myarrow] (walk)  -- (classification);	

        \draw[myarrow] (lfr)              -| (nmi);	
        \draw[myarrow] (classification)   -| (nmi);	

        \node[right=1cm of nmi] (dummy1) {}; 
        \draw[myarrow] (nmi)   -- (dummy1);	

        \draw [myarrow, dashed, bendedarrow=-2.5cm] (classification) to (seeds);
        \node[left=1cm of walk] (dummy2) {}; 
        \node[mylabel, rotate=90] at (dummy2) (label1) {Iteration};  

    \end{tikzpicture} 
    \caption{Pipeline}
    \label{fig:pipeline}
\end{figure}

Our experimental setup consists of five parts (see Figure~\ref{fig:pipeline}) but the 
respective parts differ slightly depending on whether we test non-overlapping or overlapping 
communities. We use the LFR benchmark graph generator developed by Lancichinetti, 
Fortunato, and Radicchi~\cite{LFR08, LF09}, which outputs graphs where the community 
information of each node is known. From each community in the graph thus generated, 
we pick a fixed number of seed nodes per community and give these as input to our algorithm. 
Once the algorithm outputs the affinities of all non-seed nodes, we classify them into 
communities and finally compare the output with the ground truth using 
normalized mutual information (NMI) as a metric~\cite{DDDA05}. We implemented 
our algorithm in \CPP\ and Python and the code is available online.\footnote{At 
\texttt{https://github.com/somnath1077/CommunityDetection}}

\paragraph{LFR.}
The LFR benchmark was designed by Lancichinetti, Fortunato and Radicci~\cite{LFR08}
generates random graphs with community structure. The intention was to establish a 
standard benchmark suite for community detection algorithms. Using this benchmark they 
did a comparative analysis of several well-known algorithms
for community detection~\cite{LF09}. To the best of our knowledge, this study seems to be the 
first where standardized tests were carried out on such a range of community detection algorithms. 
Subsequently, there has been another comprehensive study on overlapping community detection 
algorithms~\cite{XKS13} which also uses (among others) the LFR benchmark. As such, we chose this 
benchmark for our experiments and set the parameters in the same fashion as in~\cite{LF09}. 

We briefly describe the major parameters that the user has to supply 
for generating benchmark graphs in the LFR suite. The node degrees and the 
community sizes are distributed according to power law, with different exponents. 
An important parameter is the \emph{mixing parameter~$\mu$} which is the fraction of neighbors 
of a node that do not belong to any community that the node belongs to, averaged over all nodes.
The other parameters include maximum node degree, average node degree, 
minimum and maximum community sizes. For generating networks with overlapping communities, 
one can specify what fraction of nodes are present in multiple communities.

In what follows, we describe tests for non-overlapping and overlapping communities separately, since 
there are several small differences in out setup for these two cases. 

\subsection{Non-overlapping communities}
The networks we test have either 1000 nodes or 5000 nodes. The average node degree
was set at 20 and the maximum node degree set at 50. The parameter controlling the 
distribution of node degrees was set at~2 and that for the community size distribution was 
set at~1. Moreover, we distinguished between big and small communities: small communities have 
10--50 nodes and large communities have 20--100 nodes. 
For each of the four combinations of network and community size, we generated graphs with the 
above parameters and with varying mixing parameters. For each of these graphs, we tested the 
community information output by our algorithm and compared it against the ground truth 
using the normalized mutual information as a metric. The plots in the next section 
show how the performance varies as the mixing parameter was changed. Each data point in 
these plots is the average over 100 iterations using the same parameters. 

\paragraph{Seed node generation.} 
To use our algorithm, we expect that users pick seed nodes from 
every community that they wish to identify in the network. 
We simulate this by picking a fixed fraction of nodes from each community as seed nodes.
One of our assumptions is that the user knows the more important members of each community. 
To replicate this phenomenon in our experiments, we picked a node as seed node
with a probability that is proportional to its degree.
That is, nodes with a higher degree were picked in preference to those with a lower degree.
For those nodes which were picked as seed nodes, we set the affinity to a community to be 1 if 
and only if the node belongs to that community and 0~otherwise.

\paragraph{Classification into communities.}
The input to the algorithm consists of the network, the set of seed nodes together with their 
affinities. Once the algorithm calculates the affinities of all non-seed nodes, we classify 
them into their respective communities. This is quite easy for non-overlapping 
communities where we simply assign each node to the community to which it has the 
highest affinity, breaking ties arbitrarily.

\paragraph{Iteration.}
We extended the algorithm to iteratively improve the goodness of the detected communities.
The idea is that after running the algorithm once, there are certain nodes which can be classified 
into their communities with a high degree of certitude. We add these nodes to the seed node 
set of the respective community and iterate the procedure. To be precise, in the $j\th$ round, 
let $C^j_A$ be the set of nodes that were classified as community $A$ and $S^j_A$ 
be the seed nodes of community $A$. We create $S^{j+1}_A$ as follows: For a fixed $\varepsilon > 0$, 
choose $\varepsilon \cdot |C^j_A|$ nodes of $C^j_A$ that have the highest affinity to community $A$, 
and add them to $S^j_A$ to obtain $S^{j + 1}_A$. 
The factor $\varepsilon$ declares by how much the set of seed nodes is allowed to grow in each iteration. 
Choosing $\varepsilon = 0.1$ gives good results. Repeating this procedure several times significantly 
improves the quality of the communities detected as measured by the NMI. Each iteration takes 
$O(k \cdot m \cdot \log n)$ time and hence the cost of running the iterative algorithm is 
the number of iterations times the cost of running it once. 

\subsection{Overlapping Communities.}
The LFR benchmark suite can generate networks with an overlapping community structure. 
In addition to the parameters mentioned for the non-overlapping case, there is an additional 
parameter that controls what fraction of nodes of the network are in multiple communities. 
As in the non-overlapping case, we generated graphs with 1000 and 5000 nodes with the average
node degree set at 20 and maximum node degree set at 50. We generated graphs with two types 
of community sizes: small communities with 10--50 nodes and large communities with 20--100 nodes.
Moreover, as in~\cite{LF09}, we chose two values for the mixing factor: $0.1$ and $0.3$ 
and we plot the quality of the community structure output by the algorithm 
(measured by the NMI) against the fraction of overlapping nodes in the network.

\paragraph{Seed Generation.}
As in the case for non-overlapping communities, we experimented with a non-iterative 
and an iterative version of our approach. For the non-iterative version, the percentage 
of seed nodes that we picked were 5, 10, 15 and 20$\%$ per community, with the probability
of picking a node being proportional to its degree. For the iterative version, we used 
2, 4, 6, 8 and 10$\%$ seed nodes per community. 

\paragraph{Classification into communities.} For the overlapping case, we cannot
use the naive strategy of classifying a node to a community to which it has 
maximum affinity, since we do not even know the \emph{number} of communities a node belongs to. 
We need a way to infer this information from a node's affinity vector.

For each node, we expect the algorithm to assign high affinities to the communities 
it belongs to and lower affinities to the communities it does not belong to. 
We tried assigning a node to all communities to which it has an affinity that 
exceeds a certain threshold. This, however, did not give good results.
The following strategy worked better. 

Sort the affinities of a node in descending order and let this 
sequence be $a_1, \ldots, a_k$. Calculate the differences 
$\Delta_{1}, \ldots, \Delta_{k-1}$ with $\Delta_{j-1} := a_{j - 1} - a_j$;
let $\Delta_{\mathrm{max}}$ denote the maximum difference 
and let $i$ be the smallest index for which $\Delta_i = \Delta_{\mathrm{max}}$. We then associate 
the node with the communities to which it has the affinities $a_1, \ldots, a_i$. 
The intuition is that, while the node can have a varying affinity to the communities it belongs to, 
there is likely to be a sharp decrease in affinities for the communities that the node does 
not belong to. This is what is captured by computing the difference in affinities and then finding 
out where the first big drop in affinities occurs.

\paragraph{Iteration.}
For overlapping communities, we need to extend our strategy for iteratively improving the quality 
of the communities found. As in the non-overlapping case, after $j$ rounds, we increase the 
size of the seed node set of community~$A$ by a factor~$\varepsilon$ by adding those nodes 
which were classified to be in community~$A$ and have the highest affinity to this community. 
Let $v$ be a such a node. The classification strategy explained above might have classified~$v$ 
to be in multiple communities, say, $A_1, \dots, A_l$. In this case, we assign $v$ 
to be a seed node for communities $A, A_1, \ldots, A_l$. The running time is the number of 
iterations times the cost of running the algorithm once.

\subsection{Normalized Mutual Information}
This is an information-theoretic measure that allows us the 
compare the ``distance'' between two partitions of a finite set. Let $V$ be a finite set 
with $n$ elements and let $\mathcal{A}$ and $\mathcal{B}$ be two partitions of $V$. The probability that an 
element chosen uniformly at random belongs to a partite set $A \in \mathcal{A}$ is $n_A/n$, where $n_A$ 
is the number of elements in $A$. The Shannon entropy of the partition $\mathcal{A}$ 
is defined as:
\begin{equation}\label{eqn:shannon_entropy}
H(\mathcal{A}) = - \sum_{A \in \mathcal{A}} \frac{n_A}{n} \log_2 \frac{n_A}{n}.
\end{equation}

The mutual information of two random variables is a measure of their mutual dependence. For random 
variables $X$ and $Y$ with probability mass functions $p(x)$ and $p(y)$, respectively, and 
with a joint probability mass function $p(x, y)$, the \emph{mutual information $I(X, Y)$} 
is defined as:
\begin{equation}\label{eqn:mutual_information_rv}
I(X, Y) = \sum_{x \in \Omega(X)} \sum_{y \in \Omega(Y)} p(x, y) \log \frac{p(x, y)}{p(x) p(y)},
\end{equation}
where $\Omega(X)$ is the event space of the random variable $X$.
The mutual information of two partitions $\mathcal{A}$ and $\mathcal{B}$ 
of the node set of a graph is calculated by using the so-called ``confusion matrix'' 
$\mat{N}$ whose rows correspond to ``real'' communities and whose columns correspond 
to ``found'' communities. The entry $\mat{N}(A, B)$ is the number of nodes of community 
$A$ in partition $\mathcal{A}$ that are classified into community $B$ in partition $\mathcal{B}$. 
The mutual information is defined as:
\begin{equation}\label{eqn:mutual_information_graphs}
I(\mathcal{A}, \mathcal{B}) = 
	\sum_{A \in \mathcal{A}} \sum_{B \in \mathcal{B}} \frac{n_{A, B}}{n} 
		\log \frac{n_{A, B} / n}{ (n_A / n) \cdot (n_B / n) }.  
\end{equation}

Danon \etal~\cite{DDDA05} suggested to use a normalized variant of this measure. The 
normalized mutual information $I_N(\mathcal{A}, \mathcal{B})$ between partitions 
$\mathcal{A}$ and $\mathcal{B}$ is defined as:
\begin{equation} \label{eqn:normalized_mutual_information}
I_N(\mathcal{A}, \mathcal{B}) =  \frac{2 I(\mathcal{A}, \mathcal{B})}{H(\mathcal{A}) + H(\mathcal{B})}.
\end{equation}
The normalized mutual information takes the value~1 when both partitions are identical. If both partitions 
are independent of each other, then $I_N(\mathcal{A}, \mathcal{B}) = 0$. 

The classical notion of normalized mutual information measures the distance  between 
two \emph{partitions} and hence cannot be used for overlapping community detection. 
Lancichinetti, Fortunato, and Kert\'{e}sz~\cite{LFK09} proposed a definition of the measure for 
evaluating the similarity of covers, where a \emph{cover} of the node set of a graph 
is a collection of node subsets such that every node of the graph is in at least one set. 
Their definition of normalized mutual information is:
\begin{equation} \label{eqn:nmi_LFK}
\NMI_{\mathrm{LFK}} := 1 - \frac{1}{2} 
		\left ( \frac{H(\mathcal{A} | \mathcal{B})}{H(\mathcal{A})} + \frac{H(\mathcal{B}
				| \mathcal{A})}{H(\mathcal{B})}\right ).
\end{equation}
This definition is not exactly an extension of normalized mutual information in that the values
obtained by evaluating it on two partitions is different from what is given by normalized mutual 
information evaluated on the same pair of partitions. However in this paper we use this definition 
of NMI to evaluate the quality of the overlapping communities discovered by our algorithm. 

We note that McDaid \etal~\cite{MGH11} have extended the definition of normalized mutual 
information to covers and that for partitions, their definition corresponds to the usual definition of NMI.

\section{Experimental Results}\label{sec:experiment_results}
\newcommand{\plotwidth}{0.63\linewidth}
\newcommand{\cfinderwidth}{0.96\linewidth}
\newcommand{\otherplotswidth}{0.76\linewidth}

As in the last section, we first discuss our results for the non-overlapping case followed by 
the ones for the overlapping case.

\subsection{Non-overlapping communities}
Figures~\ref{fig:no_iter_no_overlap}, \ref{fig:iter_no_overlap}, and~\ref{fig:compare_iter_no_overlap}
show the plots that we obtained for non-overlapping communities. Figure~\ref{fig:no_iter_no_overlap}
shows tests for the non-iterative method of our algorithm with 5, 10, 15, and 20$\%$ seed nodes per 
community. 

The first observation here is that anything less than 10$\%$ seed nodes per community 
do not give good results. With a seed node percentage of 10$\%$ or more and 
a mixing factor of at most~$0.4$ we achieve an NMI above $0.9$ and can compete with \textit{Infomap}, 
which was deemed to be one the best performing algorithms on the LFR benchmark~\cite{LF09}. 
Above a mixing factor of $0.4$, our algorithm has a worse performance than \textit{Infomap} 
which, curiously enough, achieves an NMI of around 1 till a mixing factor of around 
$0.6$ after which its performance drops steeply. The drop in the performance of our algorithm 
begins earlier but is not as steep. See Figure~\ref{fig:Infomap_etal} for the performance 
of Infomap and other algorithms that were studied in~\cite{LF09}. 

Figure~\ref{fig:iter_no_overlap} shows the results for the iterative approach of 
our algorithm in the non-overlapping case. When compared with the non-iterative approach, 
one can see that even after ten iterations there is a significant improvement in 
performance (See Figure~\ref{fig:compare_iter_no_overlap}). As can be seen, typically 
with 6$\%$ seed nodes per community we obtain acceptable performance (an NMI value of 
over $0.9$ with the mixing factor of up to $0.5$).

\begin{figure}[h!]
    \centering
    \begin{subfigure}{0.5\textwidth}
    \centering
    \includegraphics[width=\plotwidth]{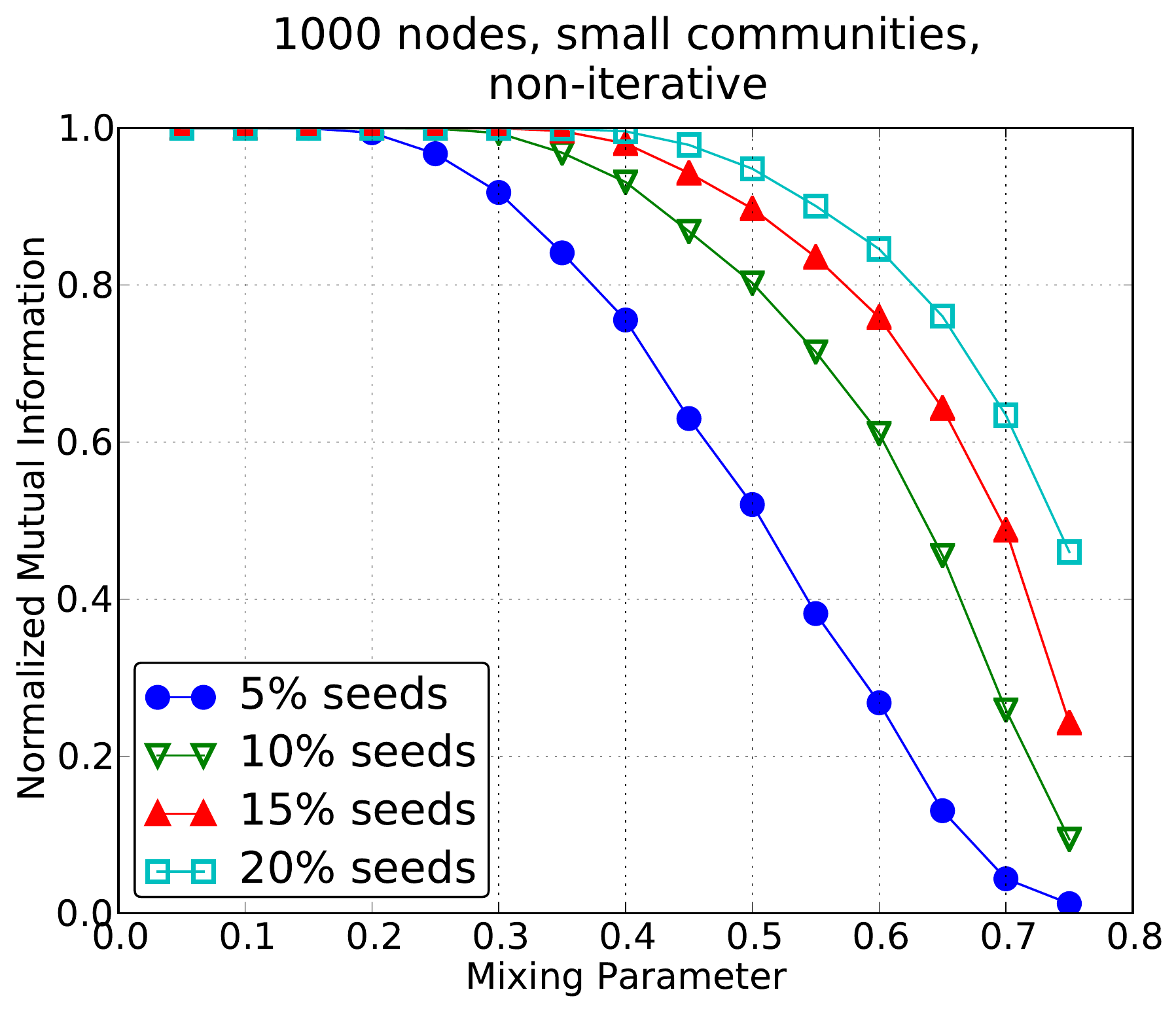}
    \end{subfigure}%
    \begin{subfigure}{0.5\textwidth}
    \centering
    \includegraphics[width=\plotwidth]{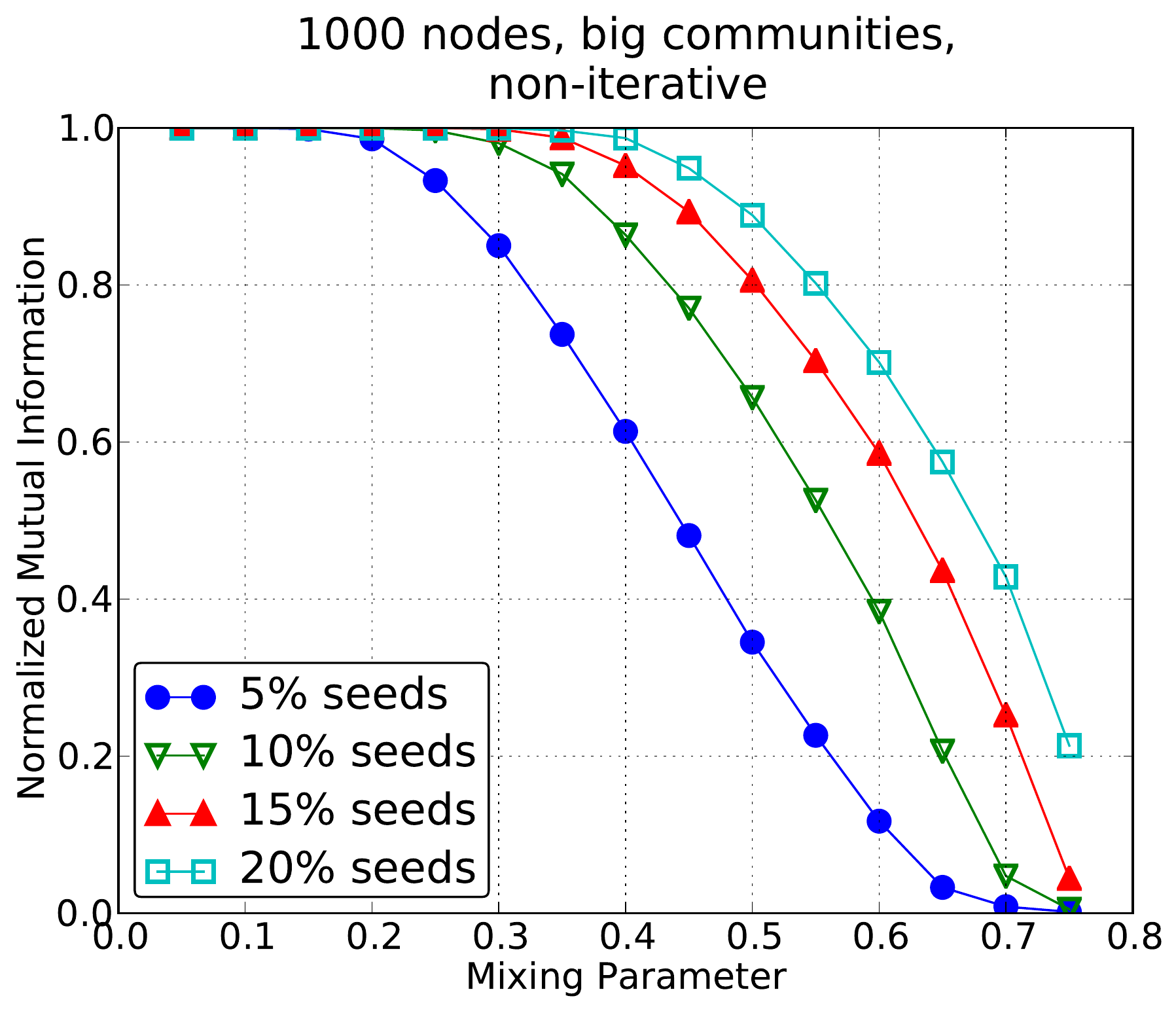}
    \end{subfigure}
    \begin{subfigure}{0.5\textwidth}
    \centering
    \includegraphics[width=\plotwidth]{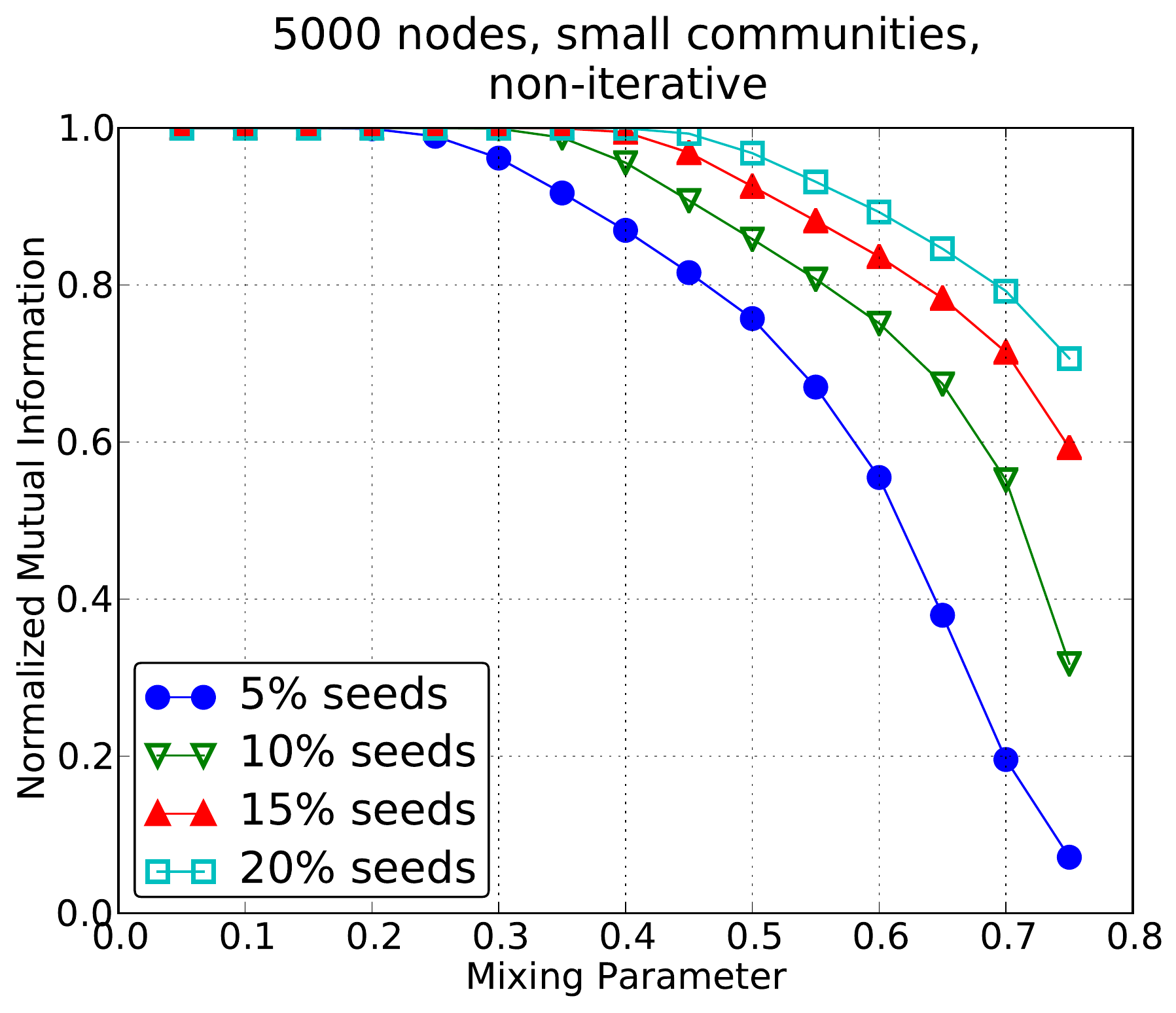}
    \end{subfigure}%
    \begin{subfigure}{0.5\textwidth}
    \centering
    \includegraphics[width=\plotwidth]{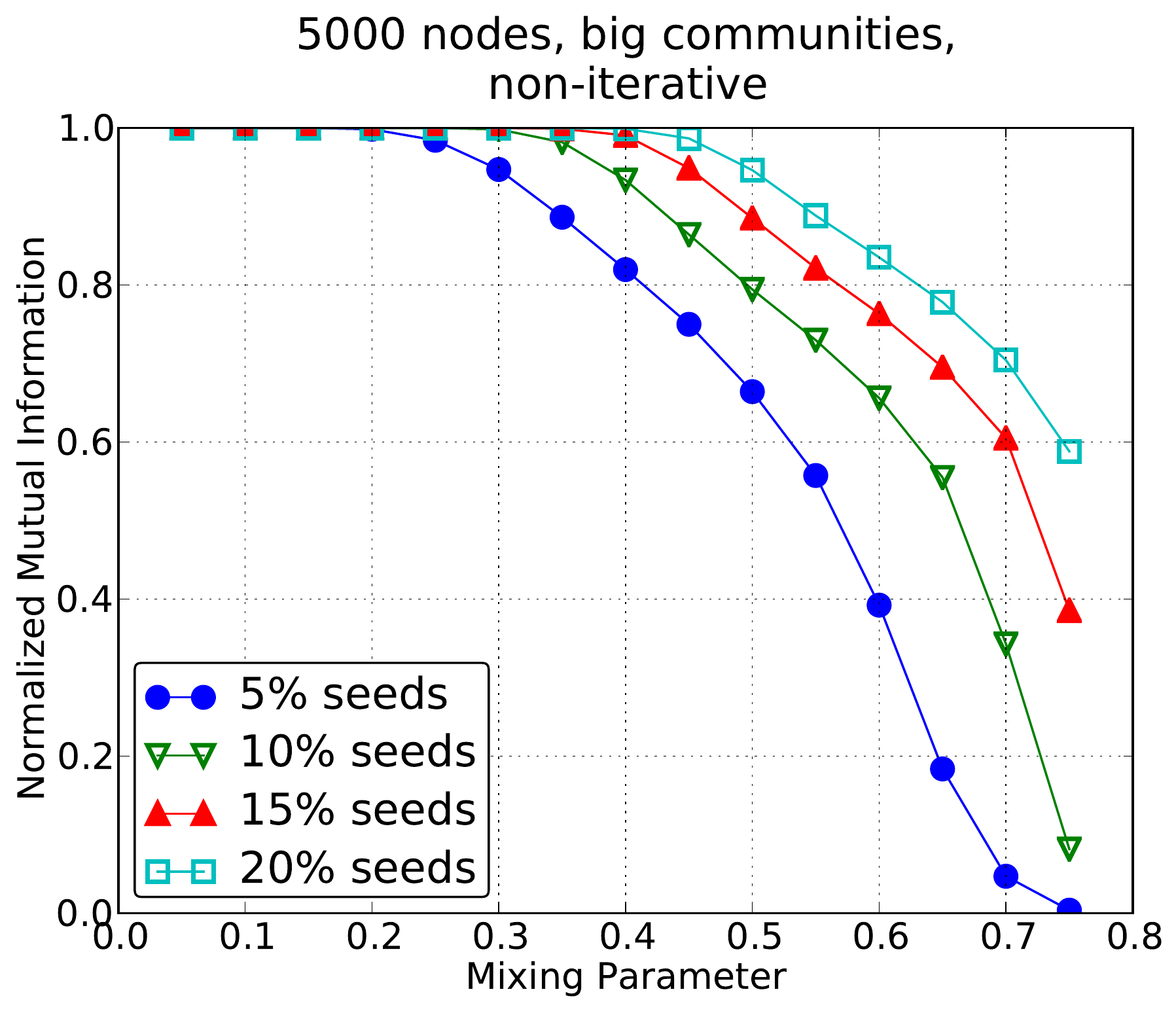}
    \end{subfigure}
    \caption{Non-iterative method for non-overlapping communities.}\label{fig:no_iter_no_overlap}
%
    \centering
    \begin{subfigure}{0.5\textwidth}
    \centering
    \includegraphics[width=\plotwidth]{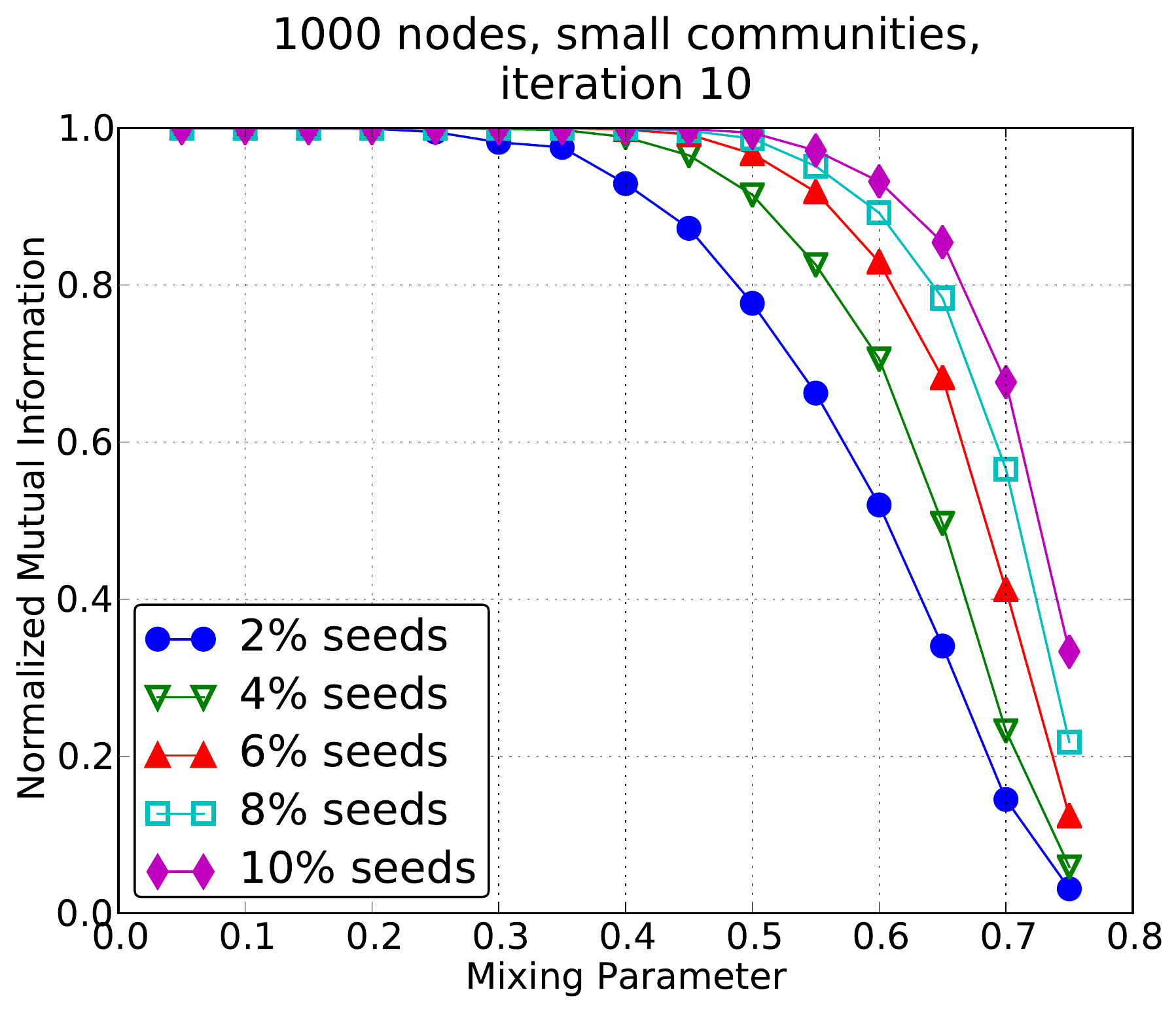}
    \end{subfigure}%
    \begin{subfigure}{0.5\textwidth}
    \centering
    \includegraphics[width=\plotwidth]{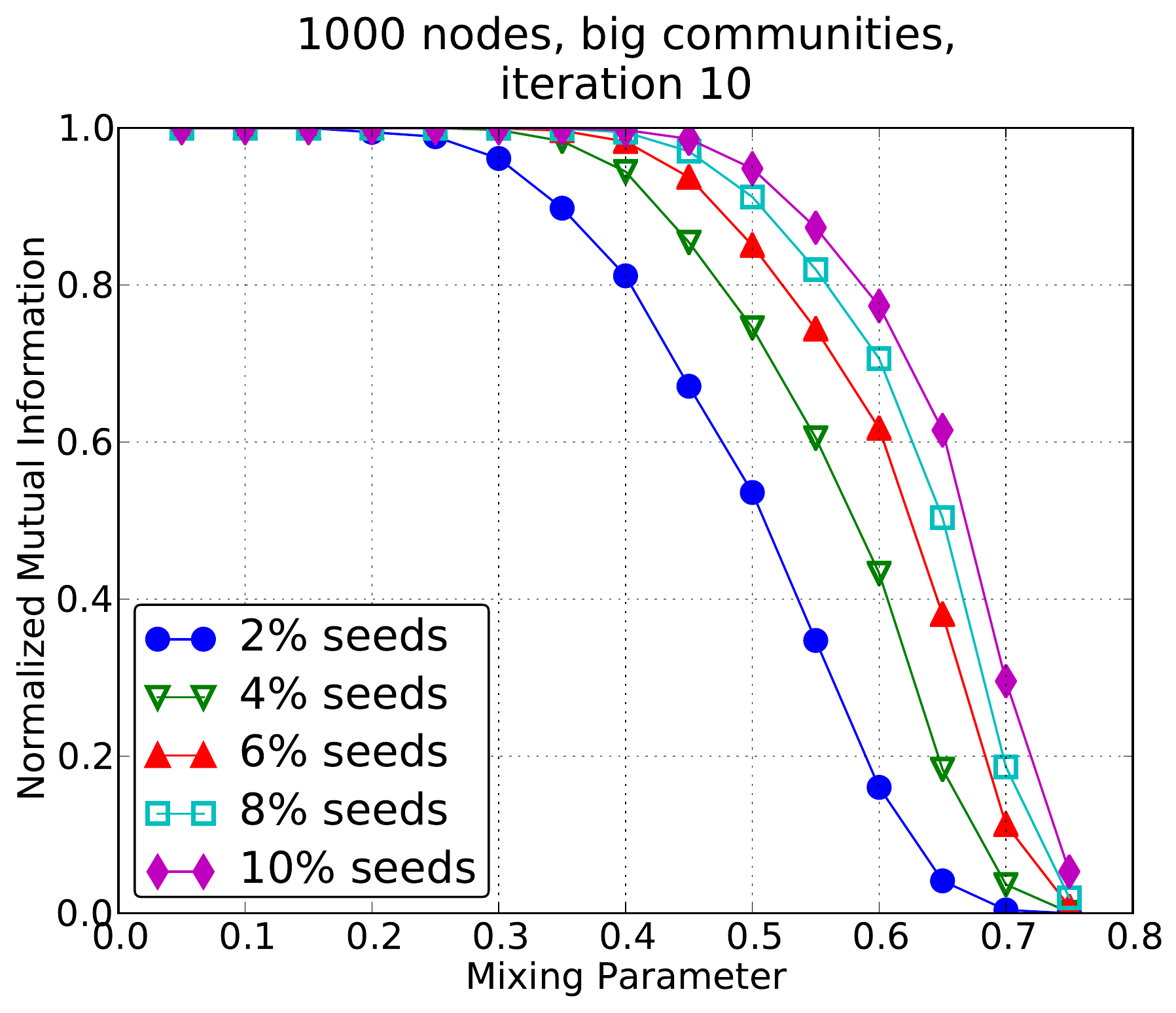}
    \end{subfigure}
    \begin{subfigure}{0.5\textwidth}
    \centering
    \includegraphics[width=\plotwidth]{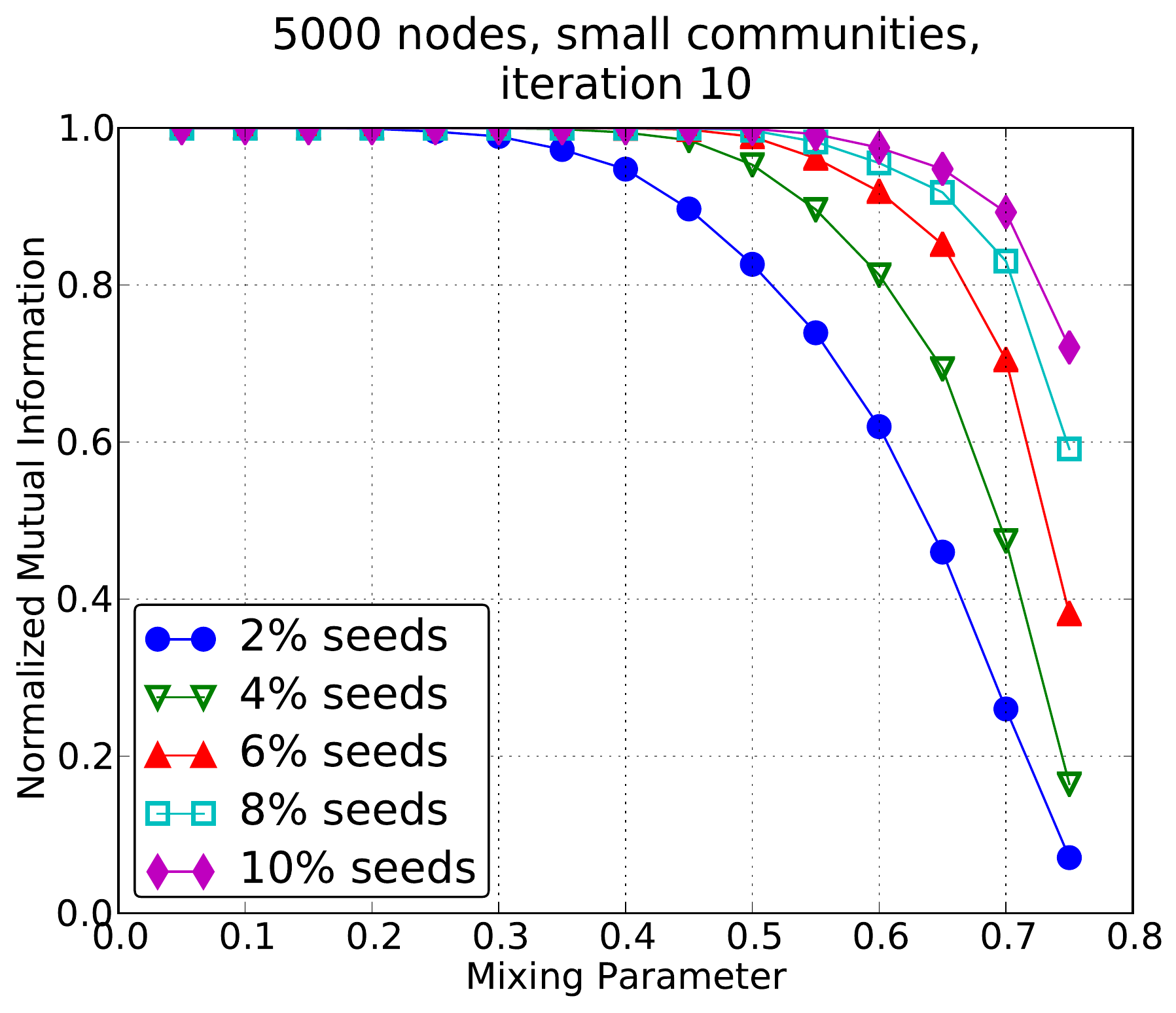}
    \end{subfigure}%
    \begin{subfigure}{0.5\textwidth}
    \centering
    \includegraphics[width=\plotwidth]{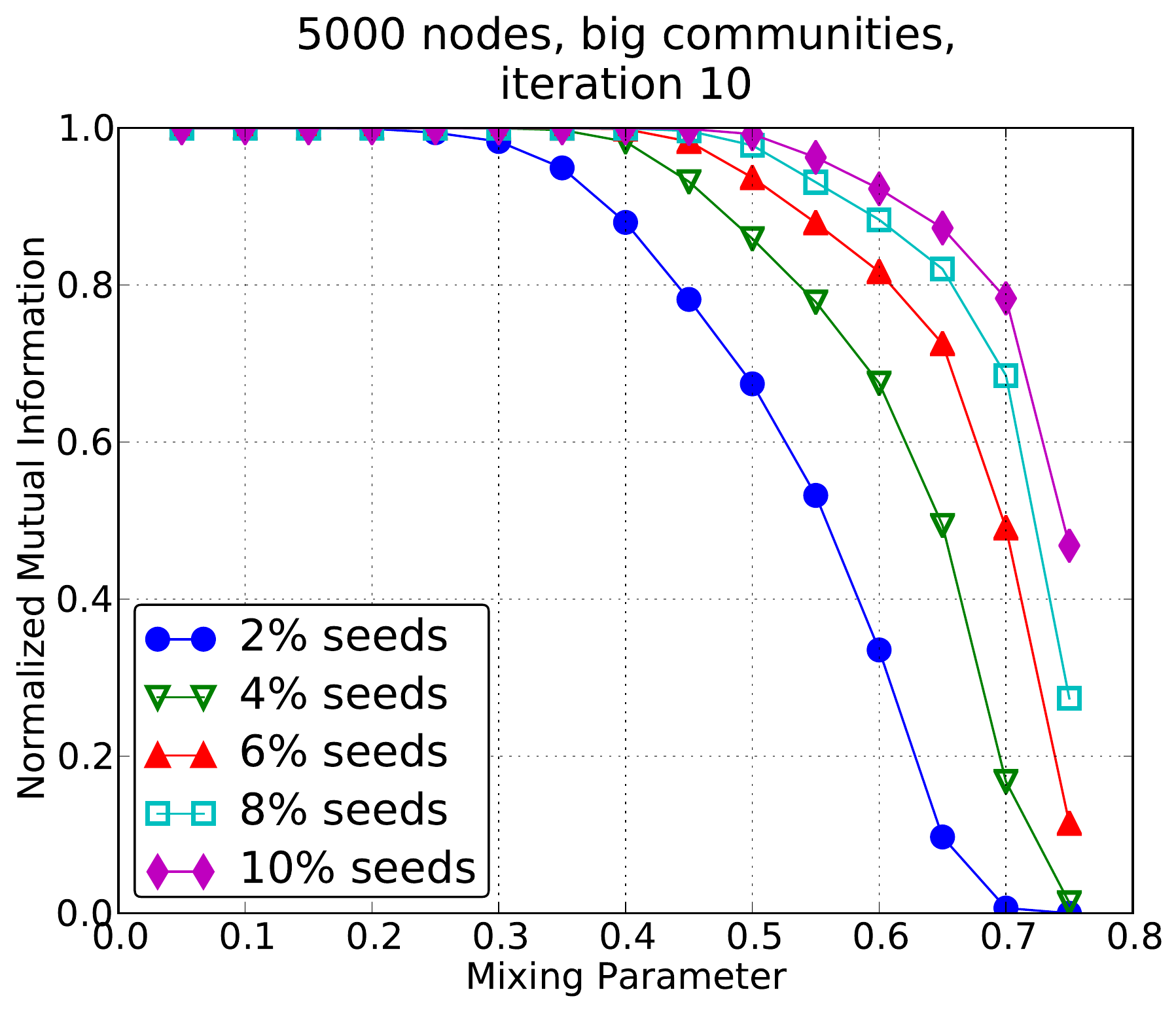}
    \end{subfigure}
    \caption{Iterative method for non-overlapping communities.}\label{fig:iter_no_overlap}
\end{figure}
\begin{figure}[h!]
    \centering
    \begin{subfigure}{0.5\textwidth}
    \centering
    \includegraphics[width=\plotwidth]{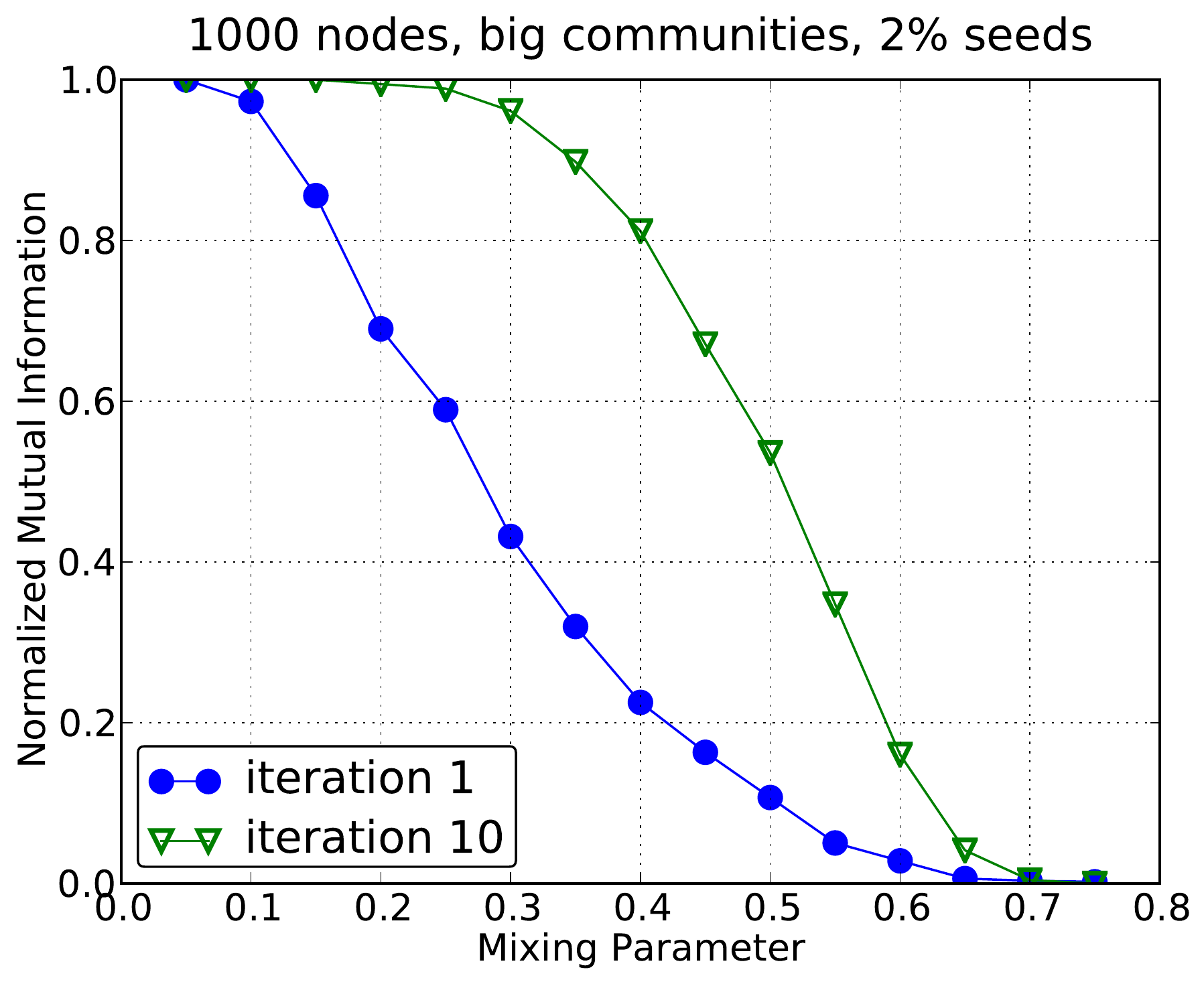}
    \end{subfigure}%
    \begin{subfigure}{0.5\textwidth}
    \centering
    \includegraphics[width=\plotwidth]{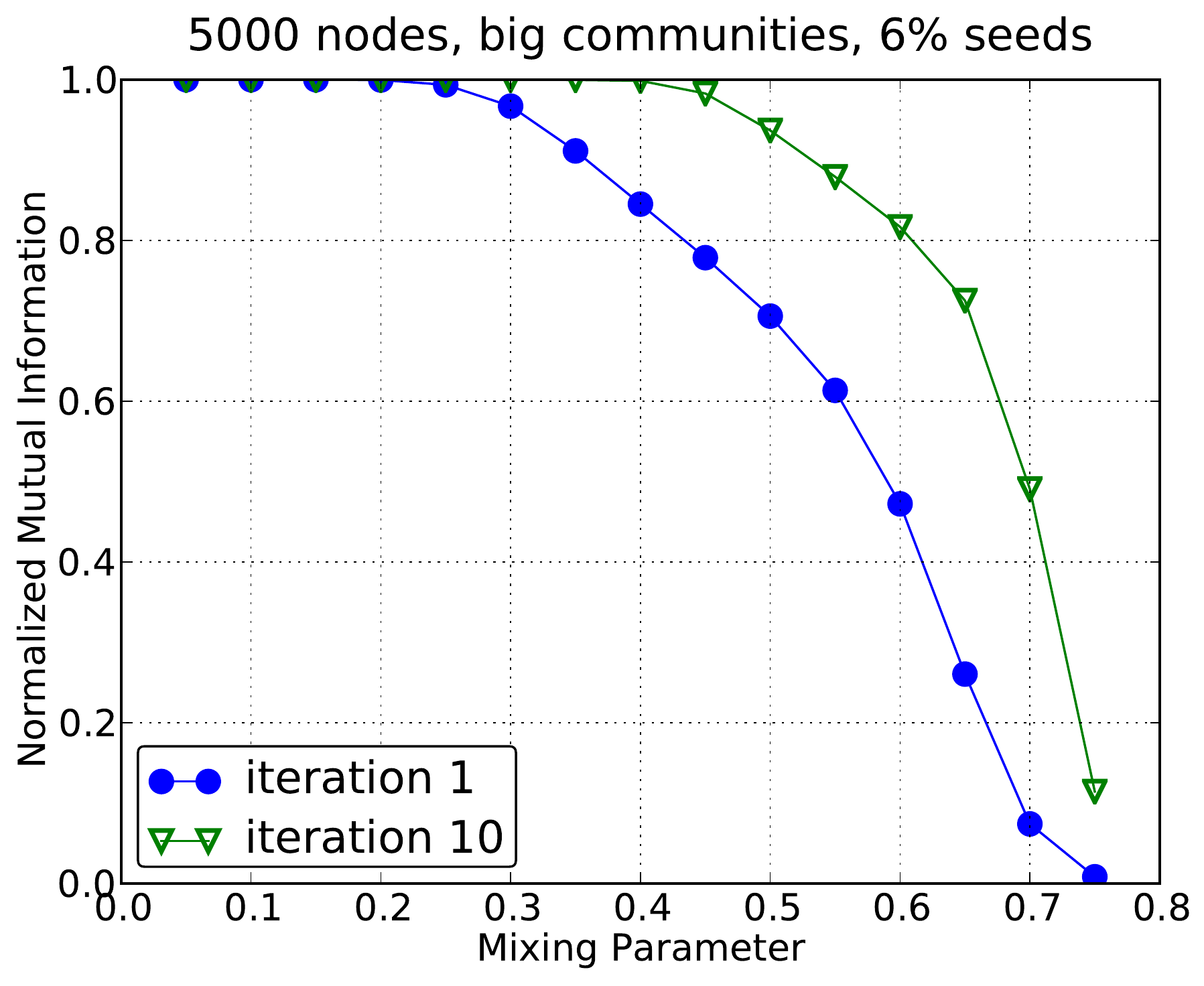}
    \end{subfigure}
    \caption{Comparison between the iterative and non-iterative method for 
		non-overlapping communities.}\label{fig:compare_iter_no_overlap}
%
    \centering
    \begin{subfigure}{0.35\textwidth}
    \centering
    \includegraphics[width=\otherplotswidth]{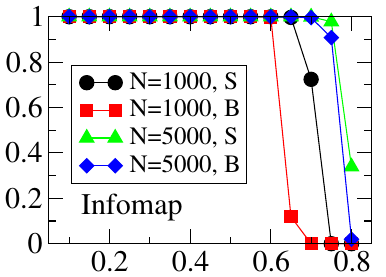}
    \end{subfigure}%
    \begin{subfigure}{0.35\textwidth}
    \centering
    \includegraphics[width=\otherplotswidth]{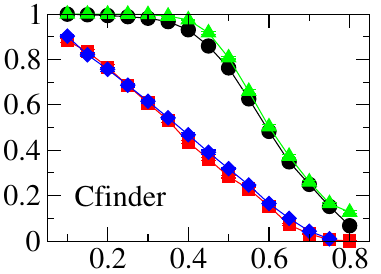}
    \end{subfigure}%
    \begin{subfigure}{0.35\textwidth}
    \centering
    \includegraphics[width=\otherplotswidth]{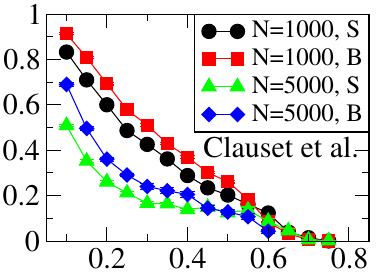}
    \end{subfigure}
    \begin{subfigure}{0.35\textwidth}
    \centering
    \includegraphics[width=\otherplotswidth]{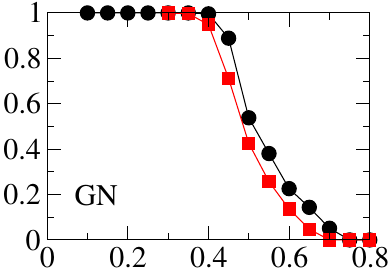}
    \end{subfigure}%
    \begin{subfigure}{0.35\textwidth}
    \centering
    \includegraphics[width=\otherplotswidth]{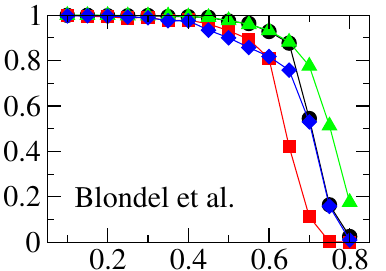}
    \end{subfigure}%
    \begin{subfigure}{0.35\textwidth}
    \centering
    \includegraphics[width=\otherplotswidth]{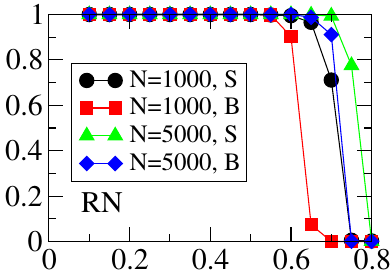}
    \end{subfigure}%
    \caption{
        Plots for Infomap, CFinder, the algorithm of Clauset \etal, Girvan-Newman (GN), Blondel \etal, 
        and the Pott's model approach by Ronhovde and Nussinov (RN) on the LFR benchmark for non-overlapping 
		communities. As usual, the NMI-value ($y$-axis) is plotted against the mixing factor ($x$-axis).
        Tests were performed on graphs with 1000 and 5000 nodes with big (B) and small (S) communities.
        Reproduced from~\cite{LF09}.
    }\label{fig:Infomap_etal}
\end{figure}

\subsection{Overlapping communities}
Figures~\ref{fig:no_iter_overlap_1000N} and~\ref{fig:no_iter_overlap_5000N} 
show our results for the overlapping case. In the study of Lancichinetti and Fortunato~\cite{LF09}, 
only one algorithm (\emph{Cfinder}~\cite{PDFV05}) for overlapping communities was benchmarked 
(see Figure~\ref{fig:CFinder_overlapping}). 
The main difference with the non-overlapping case is that typically our algorithm needs a larger 
seed node percentage per community. This is not surprising since in the overlapping case, we would 
need seed nodes from the various overlaps as well as from the non-overlapping portions of communities 
to make a good-enough calculation of the affinities. 

For graphs of both 1000 and 5000 nodes, our algorithm performs better 
than Cfinder up to an overlapping fraction of $0.4$. We stress that Cfinder 
has an exponential worst-case running time and would be infeasible on larger graphs. 
%
%

Figures~\ref{fig:iter_overlap_1000N} and~\ref{fig:iter_overlap_5000N} show the 
plots for the iterative method (with 10 iterations). A comparison of the 
non-iterative and iterative method is shown in Figure~\ref{fig:compare_iter_overlap}. 
Iteration yields an improvement in performance, as measured by the NMI, but it is 
not as dramatic as in the non-overlapping case with the NMI increase being at most 
$10\%$ at best. The percentage of seed nodes per community required in the 
iterative approach with a mixing factor of $0.3$ is around 8$\%$.

\begin{figure}[h!]
    \centering
    \begin{subfigure}{0.5\textwidth}
    \centering
    \includegraphics[width=\plotwidth]{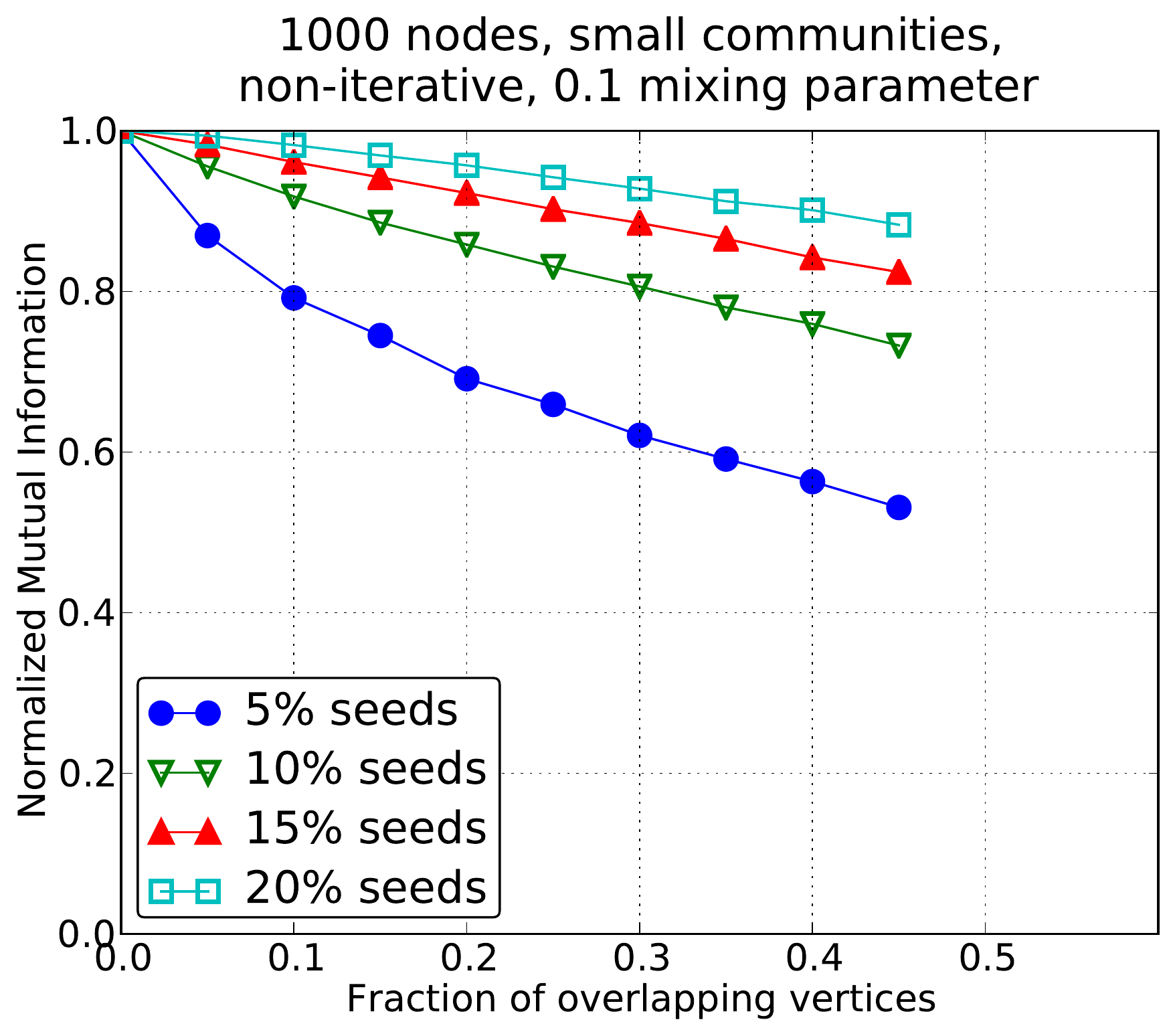}
    \end{subfigure}%
    \begin{subfigure}{0.5\textwidth}
    \centering
    \includegraphics[width=\plotwidth]{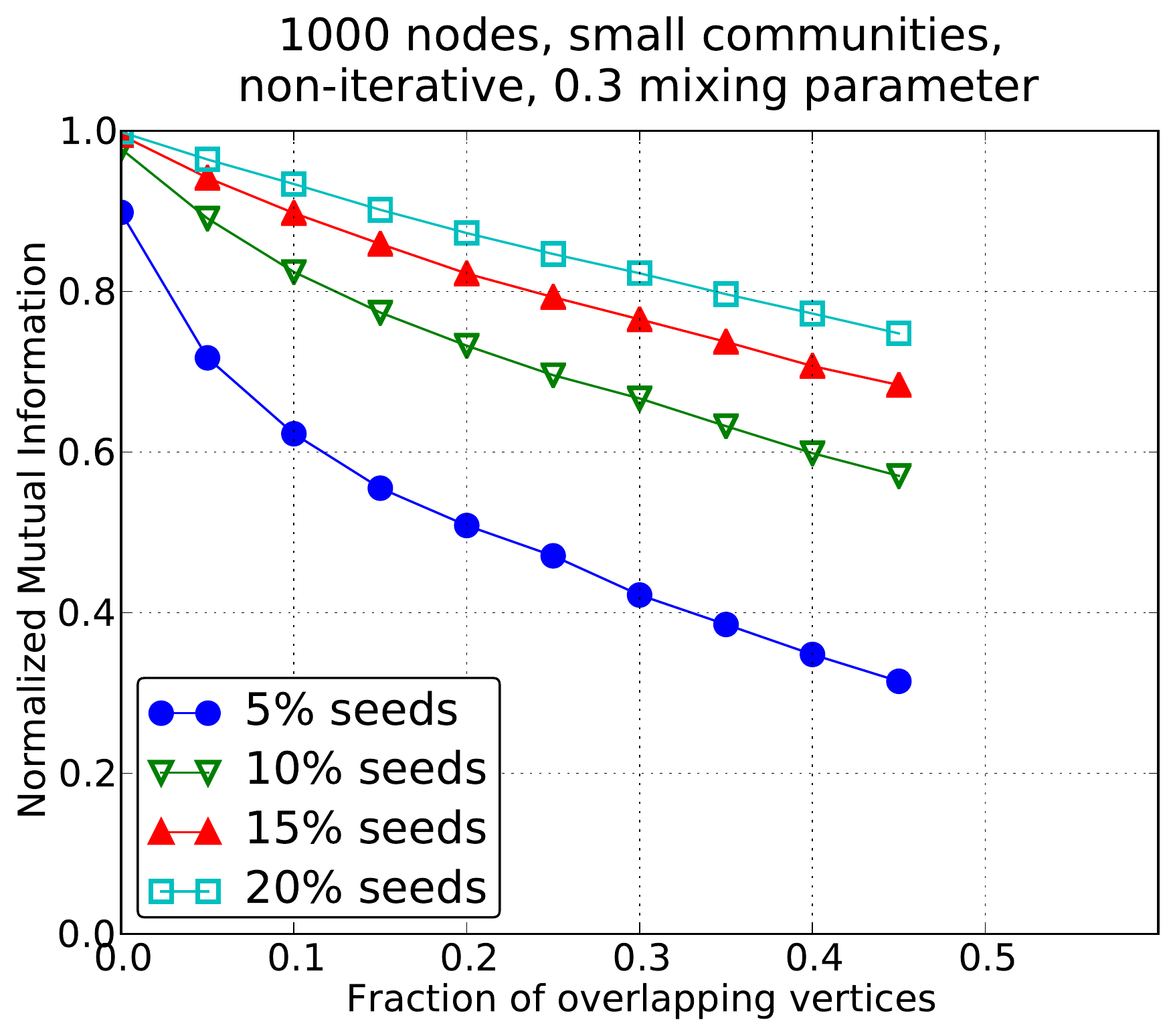}
    \end{subfigure}
    \begin{subfigure}{0.5\textwidth}
    \centering
    \includegraphics[width=\plotwidth]{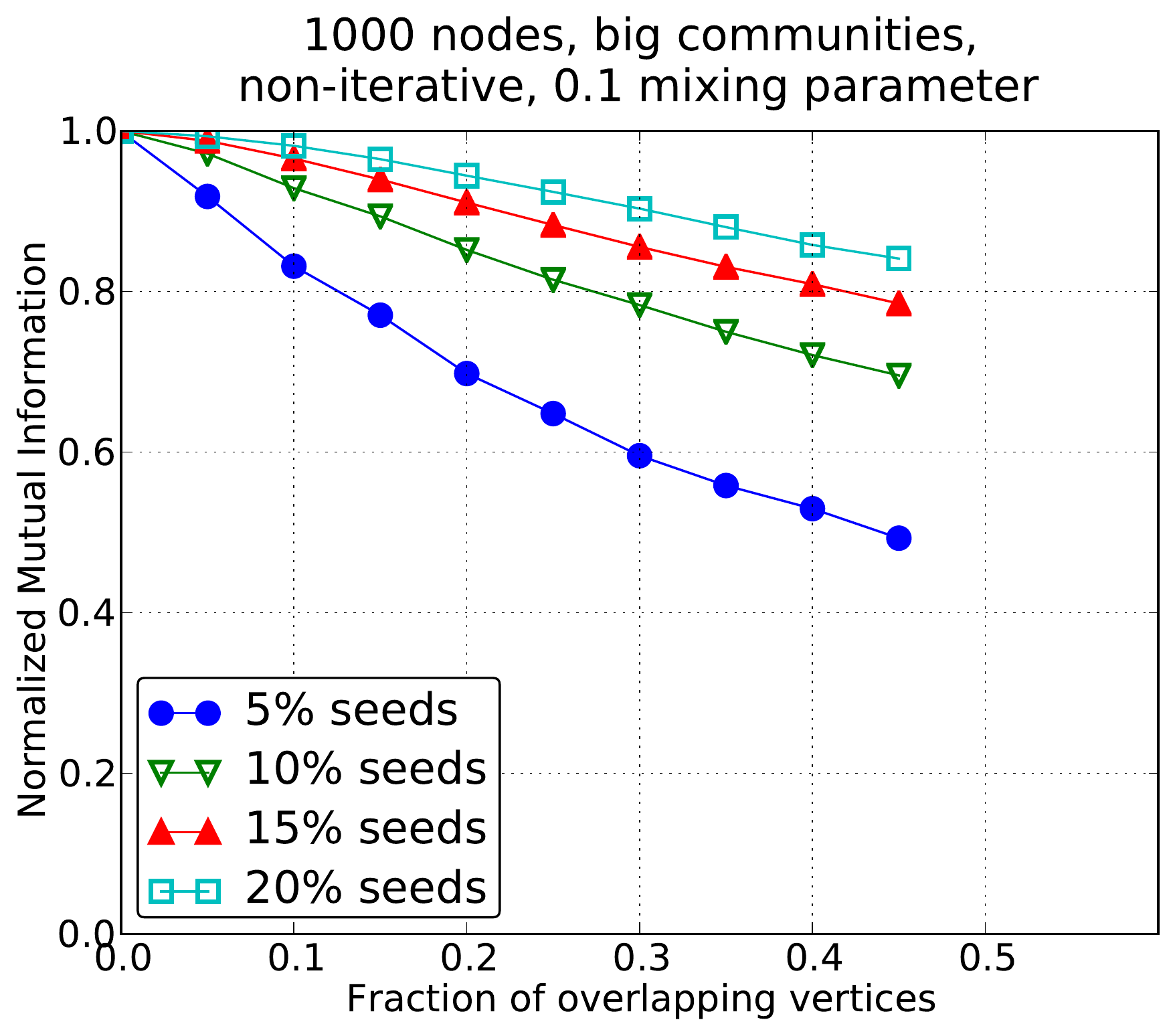}
    \end{subfigure}%
    \begin{subfigure}{0.5\textwidth}
    \centering
    \includegraphics[width=\plotwidth]{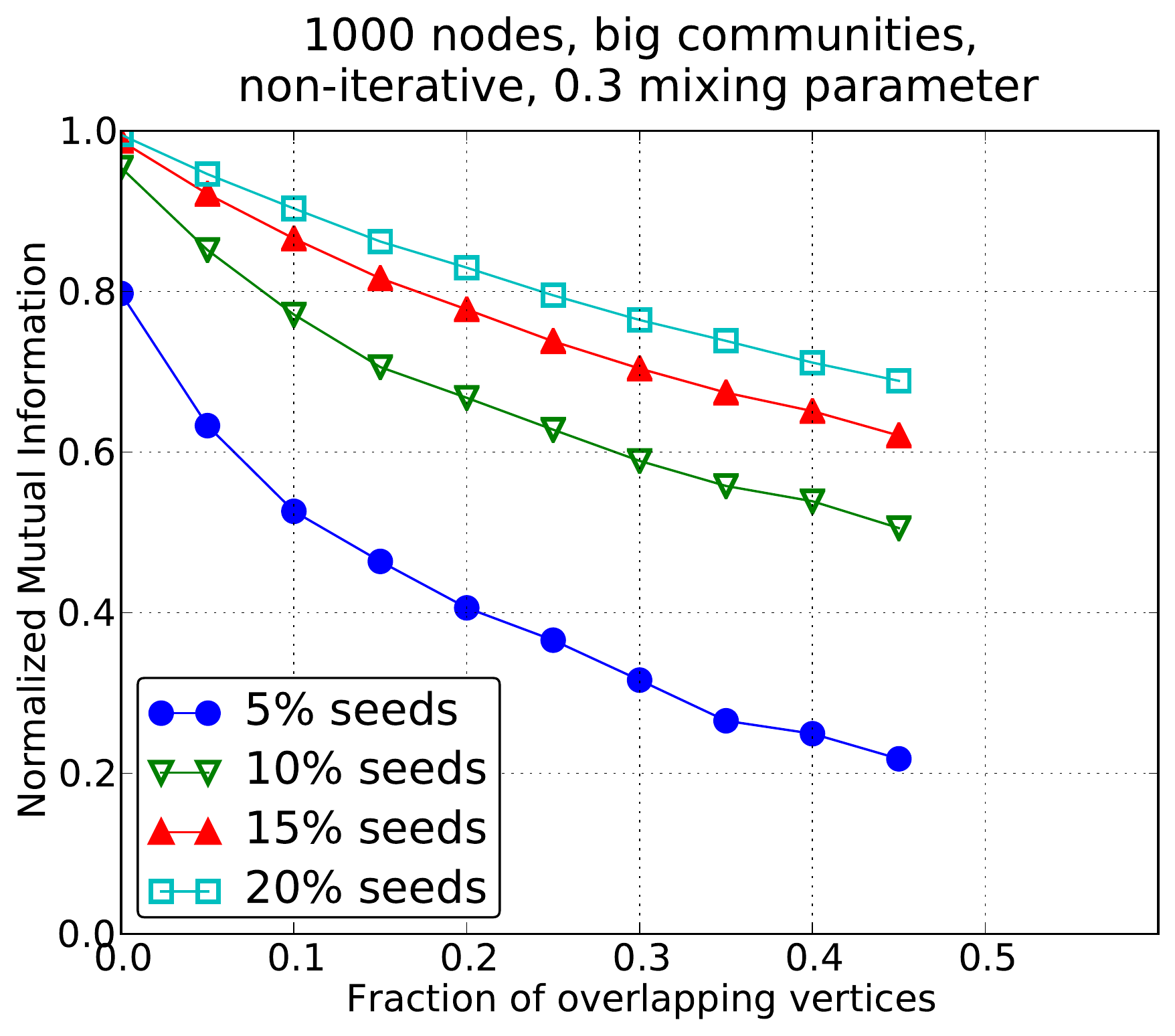}
    \end{subfigure}
    \caption{Non-iterative method for overlapping communities on 1000 nodes.}\label{fig:no_iter_overlap_1000N}
%
    \centering
    \begin{subfigure}{0.5\textwidth}
    \centering
    \includegraphics[width=\plotwidth]{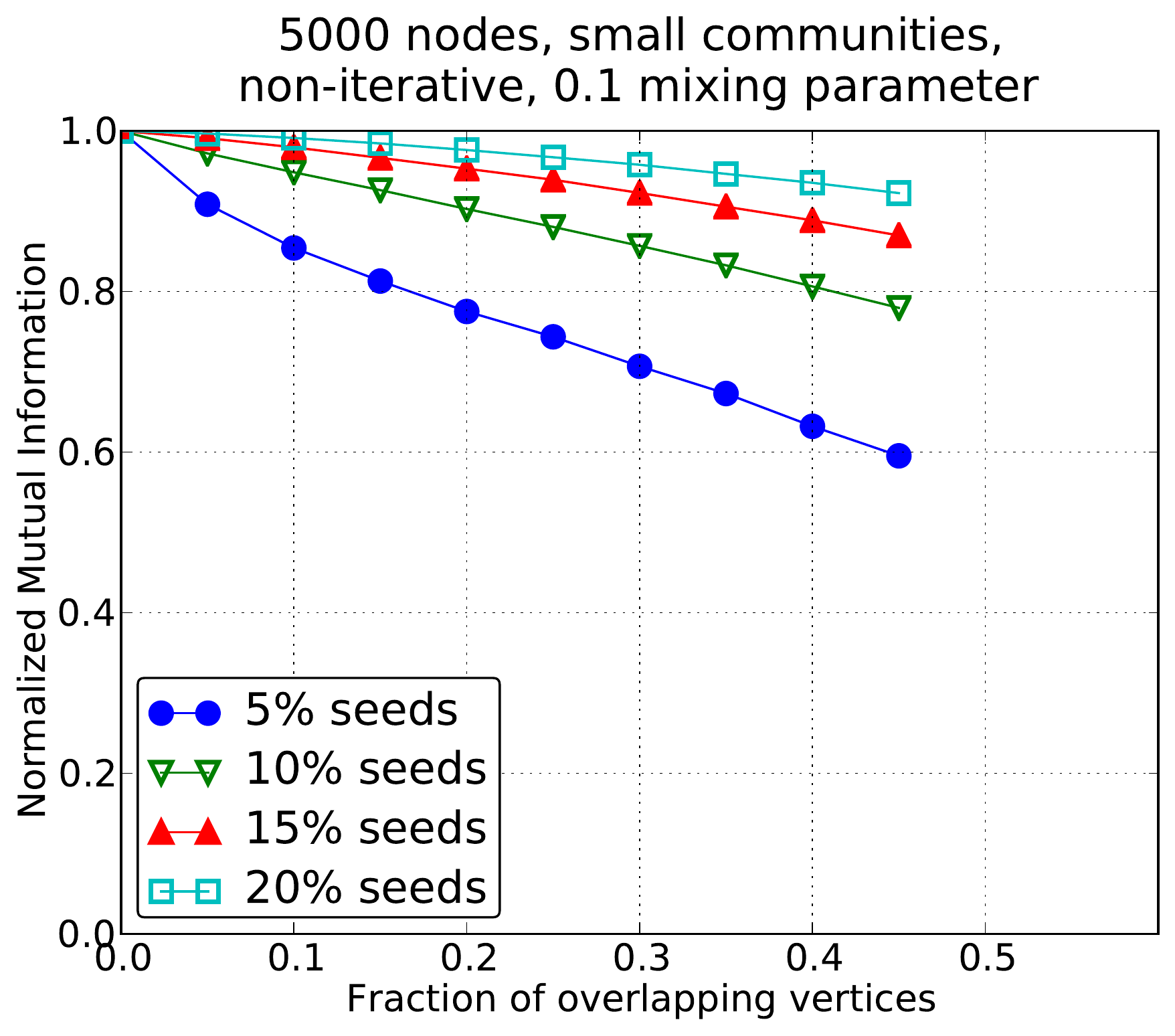}
    \end{subfigure}%
    \begin{subfigure}{0.5\textwidth}
    \centering
    \includegraphics[width=\plotwidth]{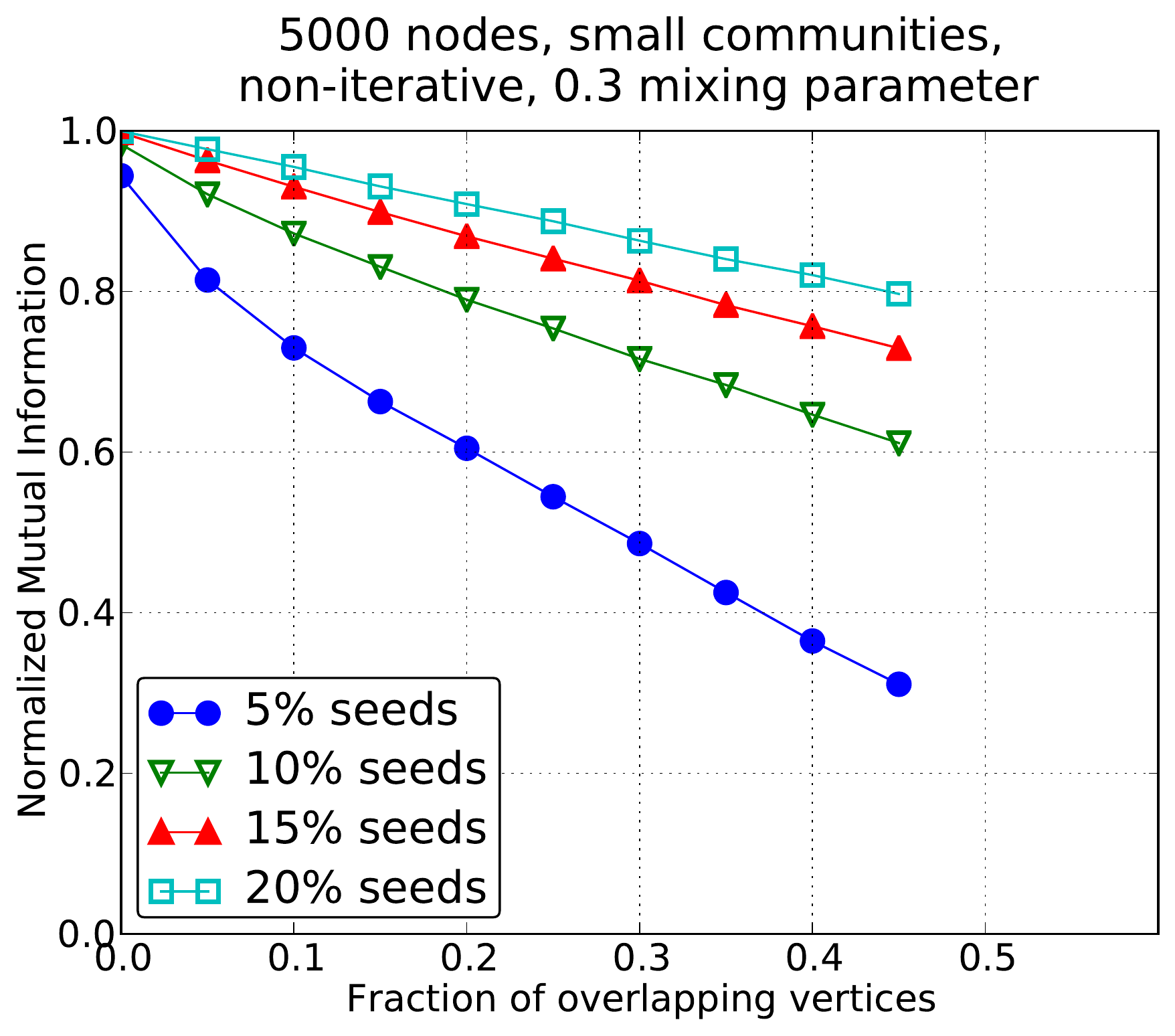}
    \end{subfigure}
    \begin{subfigure}{0.5\textwidth}
    \centering
    \includegraphics[width=\plotwidth]{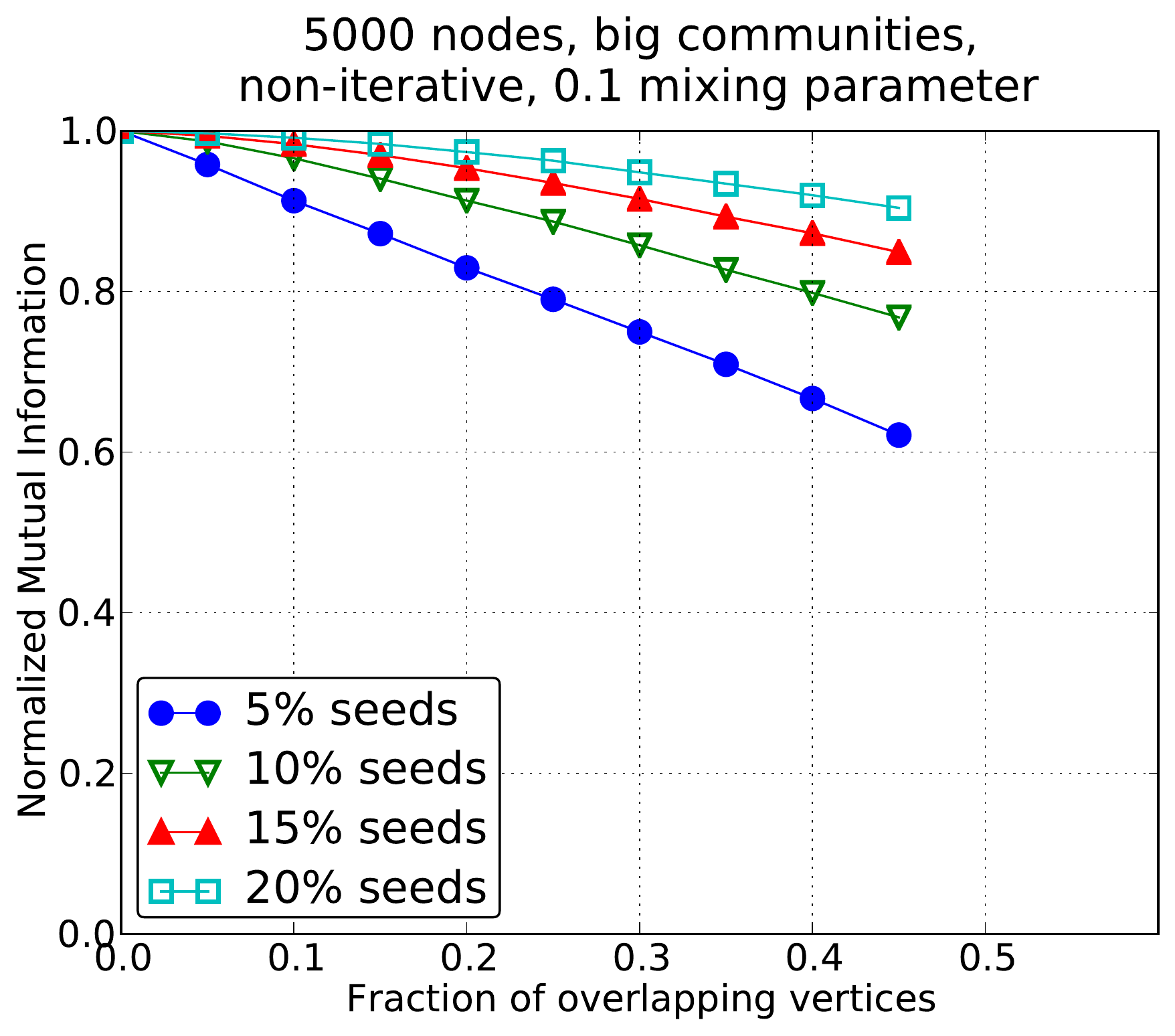}
    \end{subfigure}%
    \begin{subfigure}{0.5\textwidth}
    \centering
    \includegraphics[width=\plotwidth]{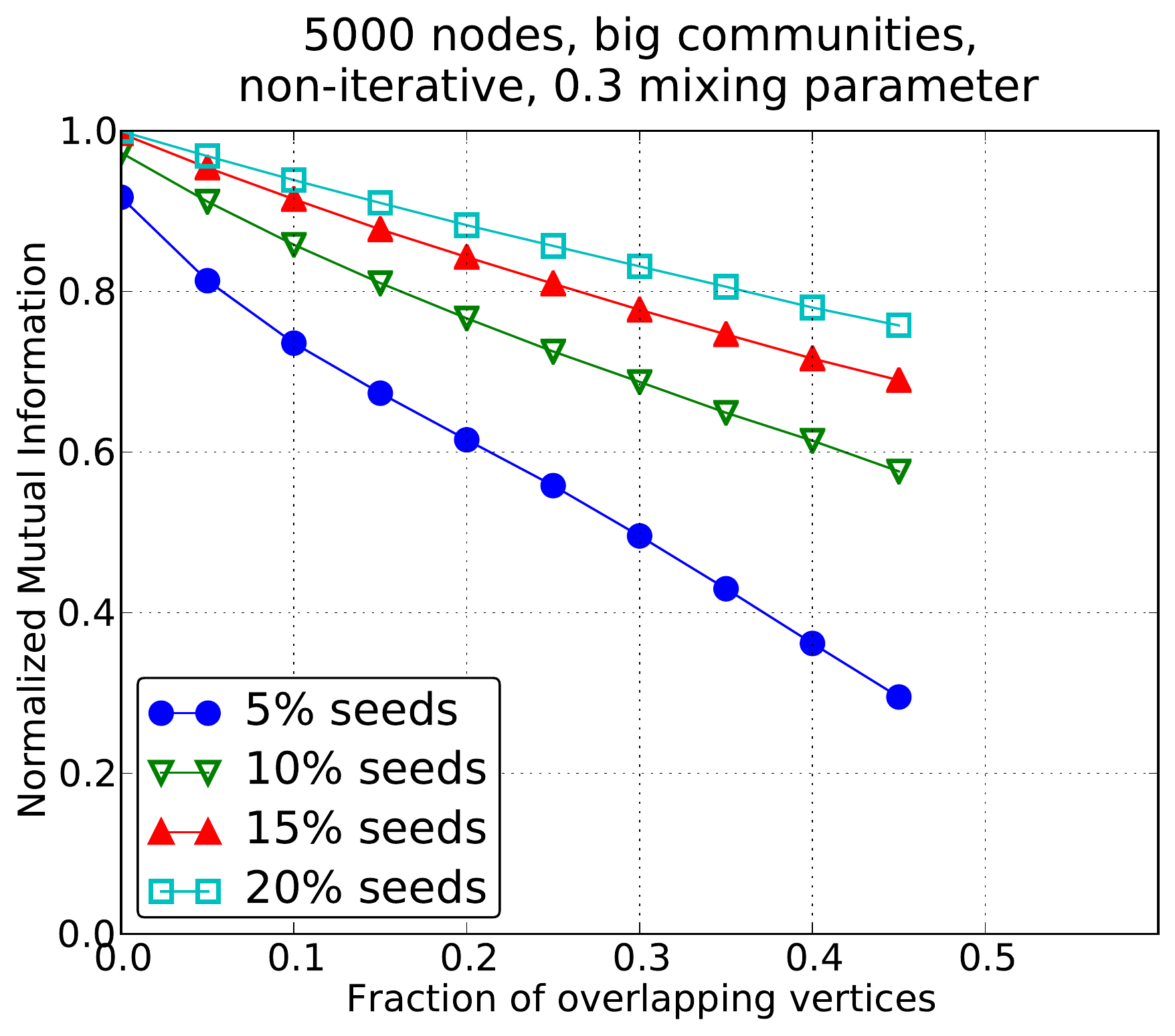}
    \end{subfigure}
    \caption{Non-iterative method for overlapping communities on 5000 nodes.}\label{fig:no_iter_overlap_5000N}
\end{figure}
\begin{figure}[h!]
    \centering
    \begin{subfigure}{0.5\textwidth}
    \centering
    \includegraphics[width=\plotwidth]{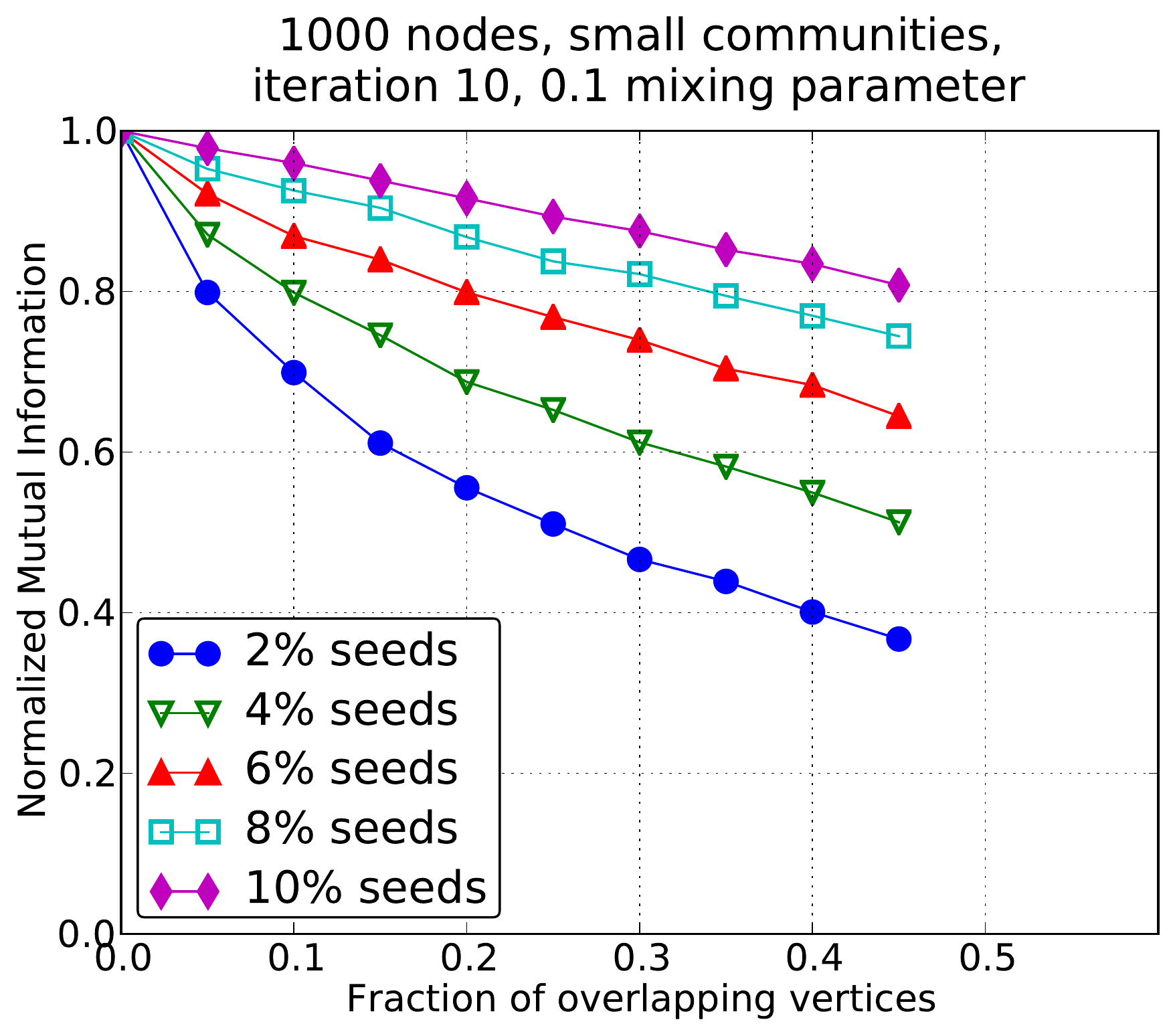}
    \end{subfigure}%
    \begin{subfigure}{0.5\textwidth}
    \centering
    \includegraphics[width=\plotwidth]{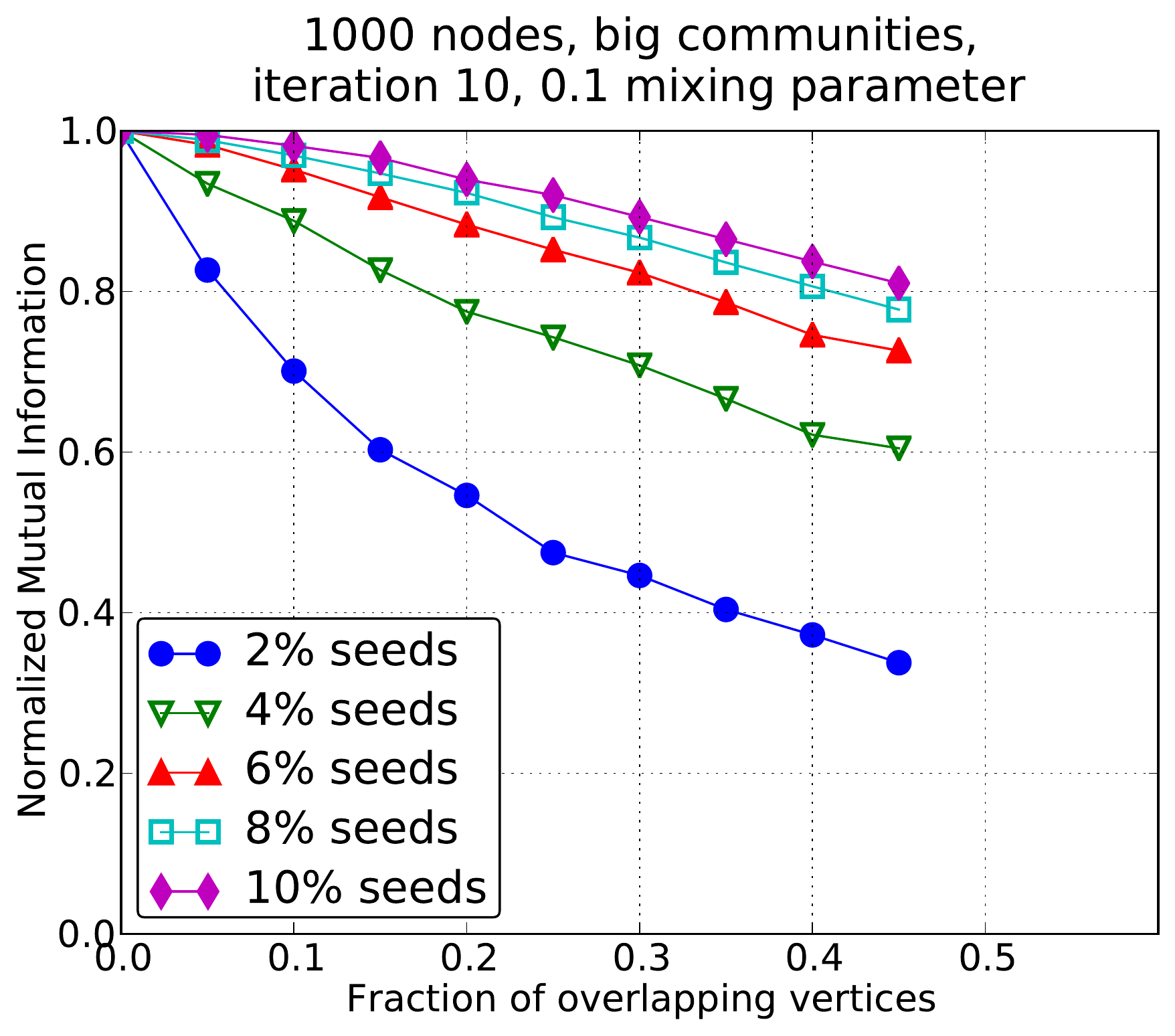}
    \end{subfigure}
    \begin{subfigure}{0.5\textwidth}
    \centering
    \includegraphics[width=\plotwidth]{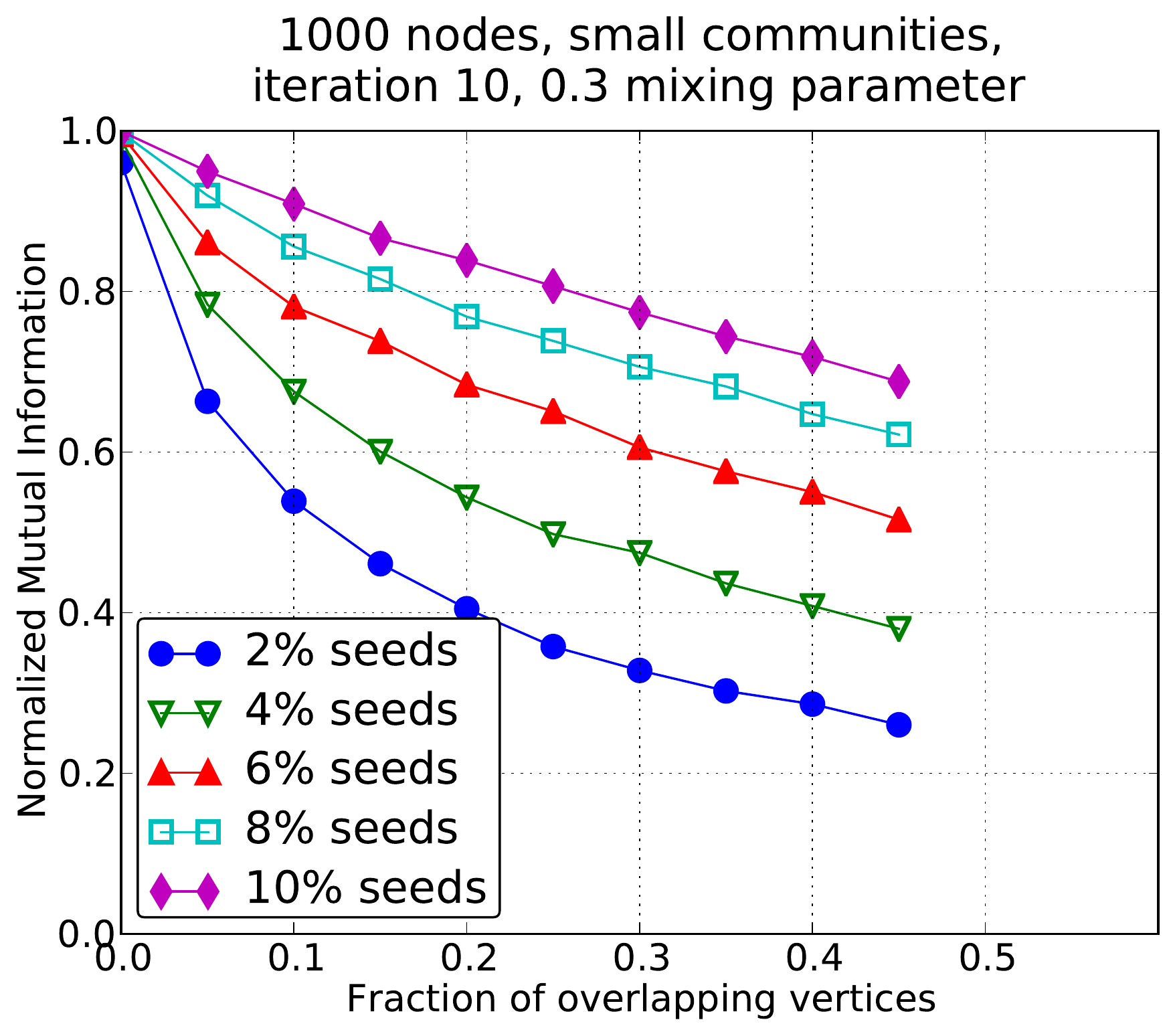}
    \end{subfigure}%
    \begin{subfigure}{0.5\textwidth}
    \centering
    \includegraphics[width=\plotwidth]{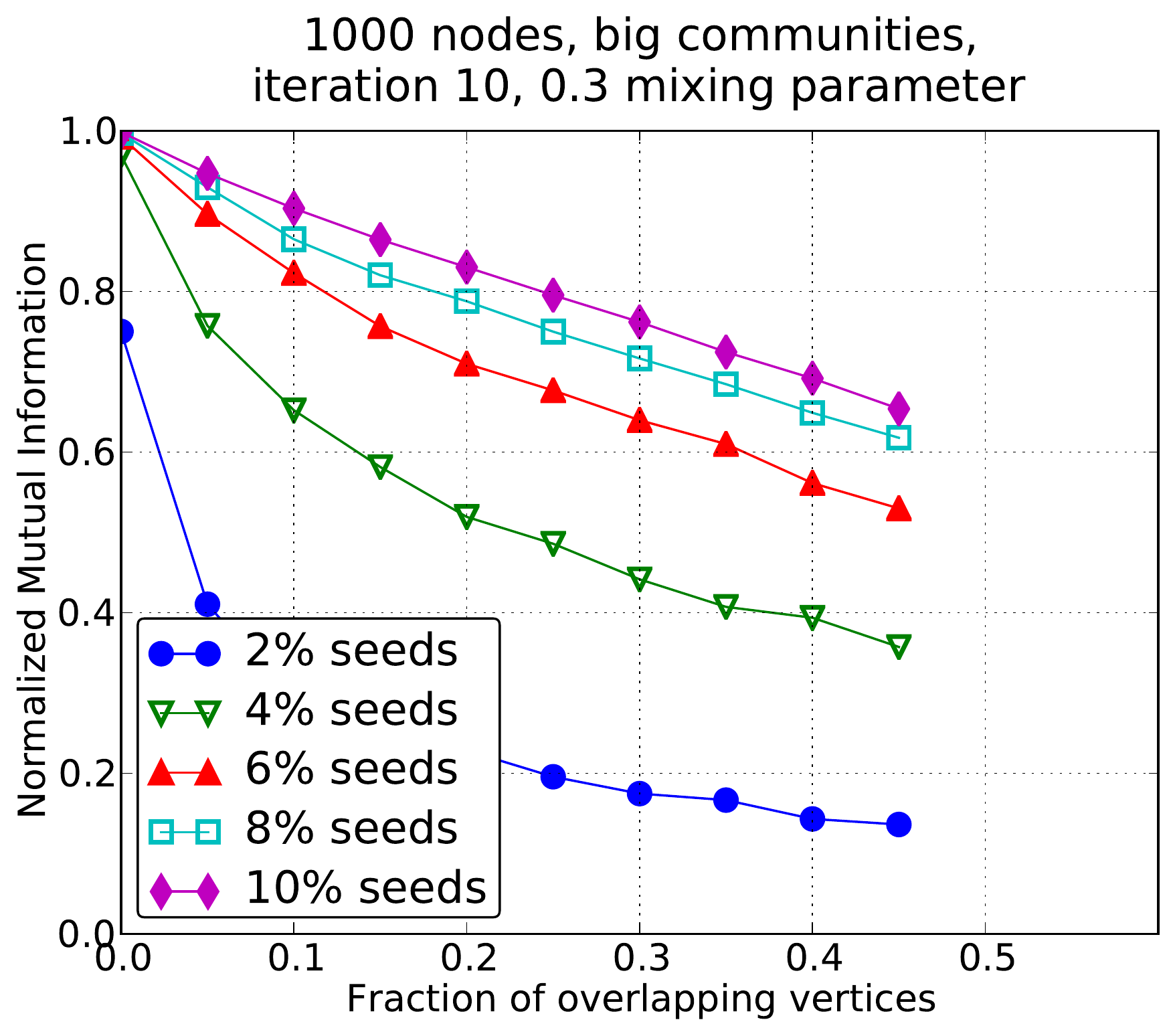}
    \end{subfigure}
    \caption{Iterative method for overlapping communities on 1000 nodes.}\label{fig:iter_overlap_1000N}
%
    \centering
    \begin{subfigure}{0.5\textwidth}
    \centering
    \includegraphics[width=\plotwidth]{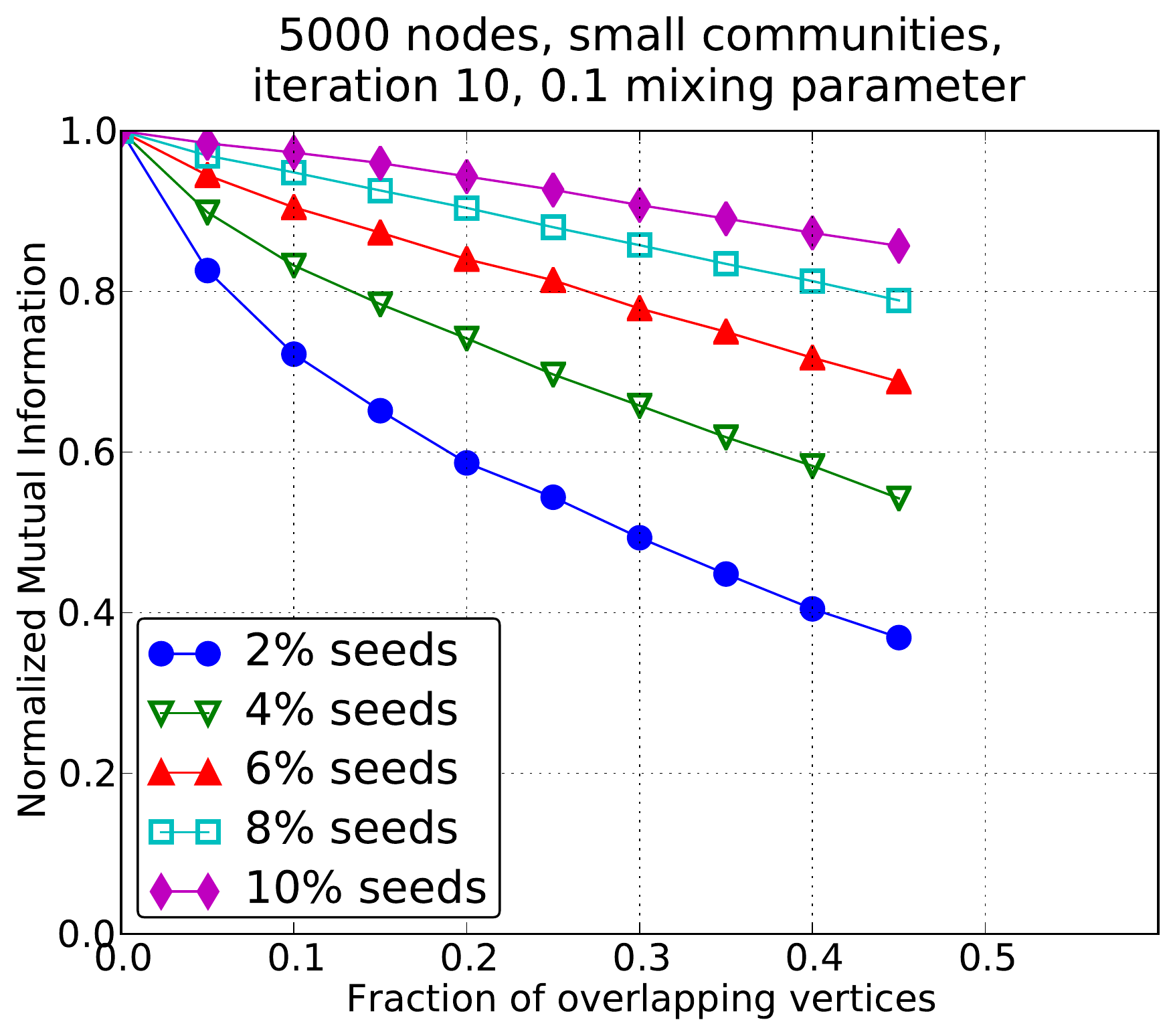}
    \end{subfigure}%
    \begin{subfigure}{0.5\textwidth}
    \centering
    \includegraphics[width=\plotwidth]{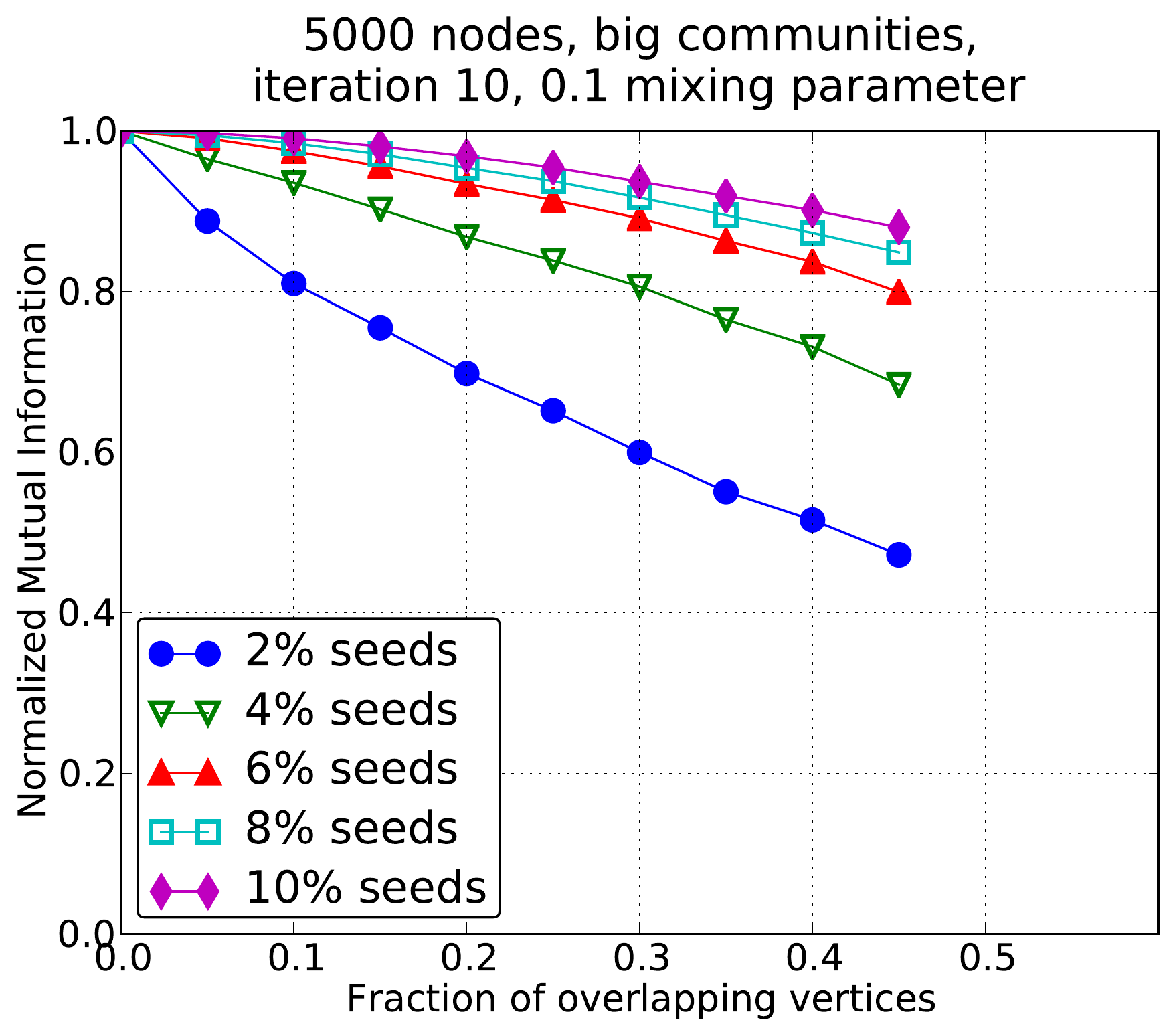}
    \end{subfigure}
    \begin{subfigure}{0.5\textwidth}
    \centering
    \includegraphics[width=\plotwidth]{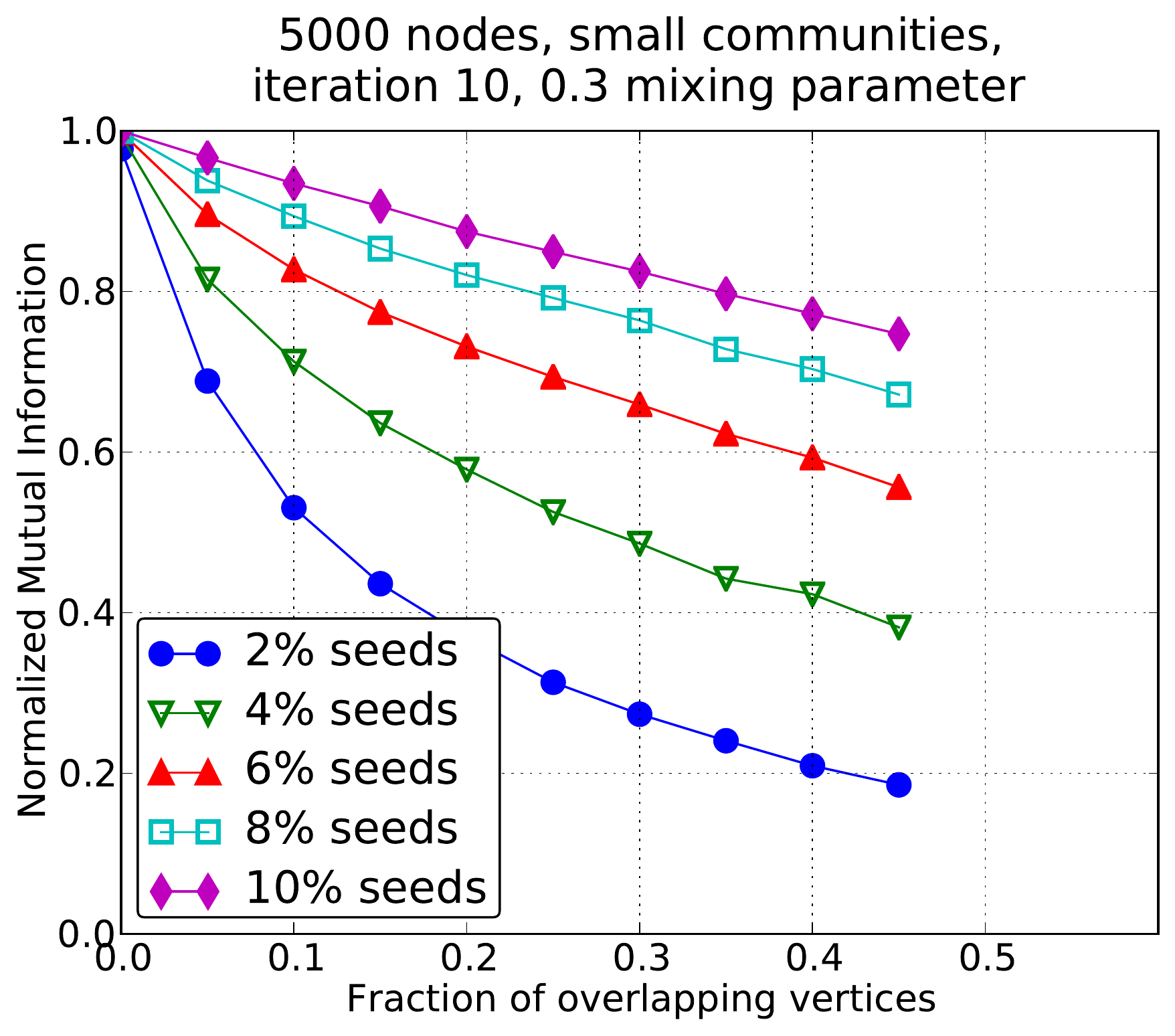}
    \end{subfigure}%
    \begin{subfigure}{0.5\textwidth}
    \centering
    \includegraphics[width=\plotwidth]{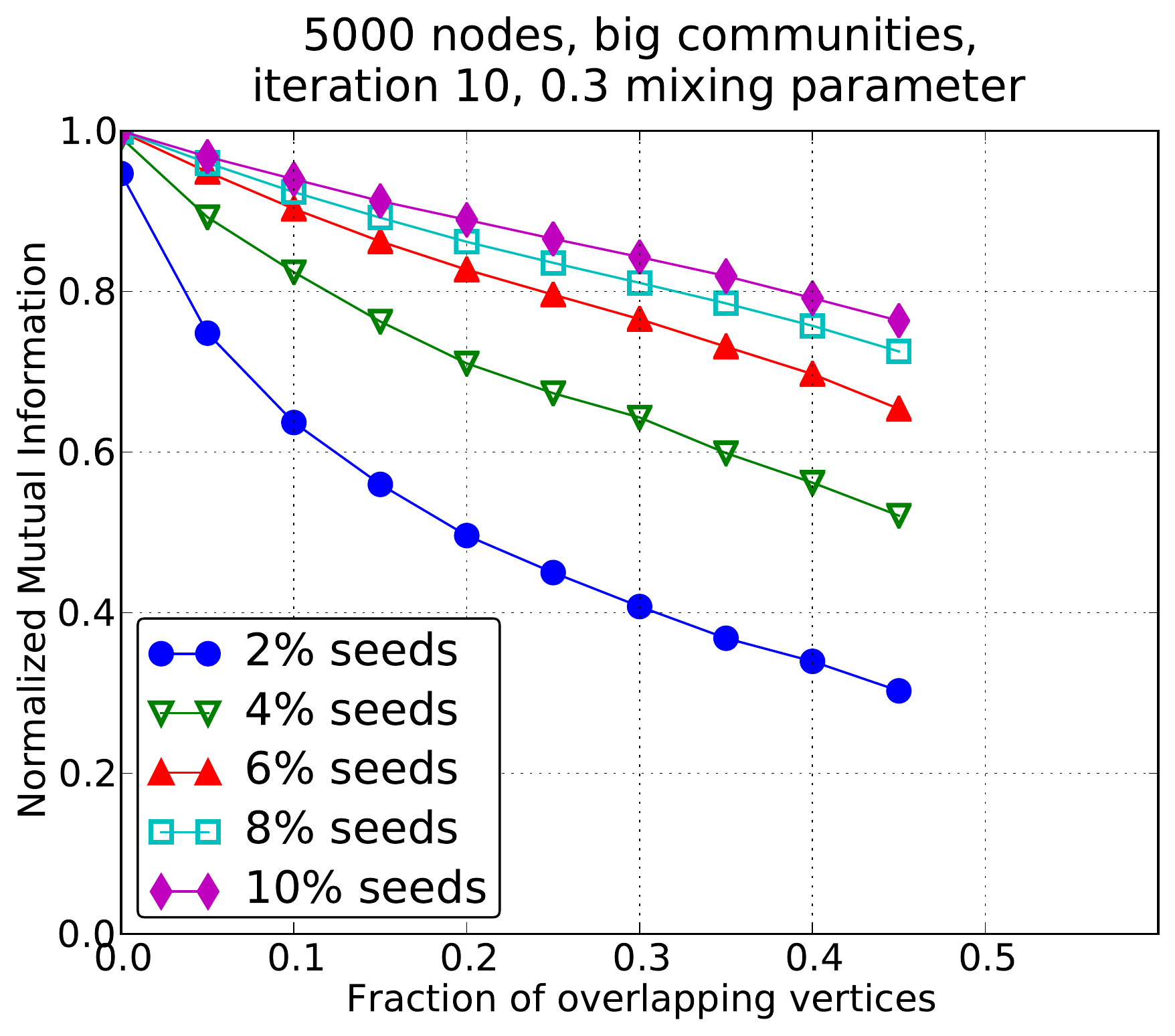}
    \end{subfigure}
    \caption{Iterative method for overlapping communities on 5000 nodes.}\label{fig:iter_overlap_5000N}
\end{figure}
\begin{figure}[h!]
    \centering
    \begin{subfigure}{0.5\textwidth}
    \centering
    \includegraphics[width=\plotwidth]{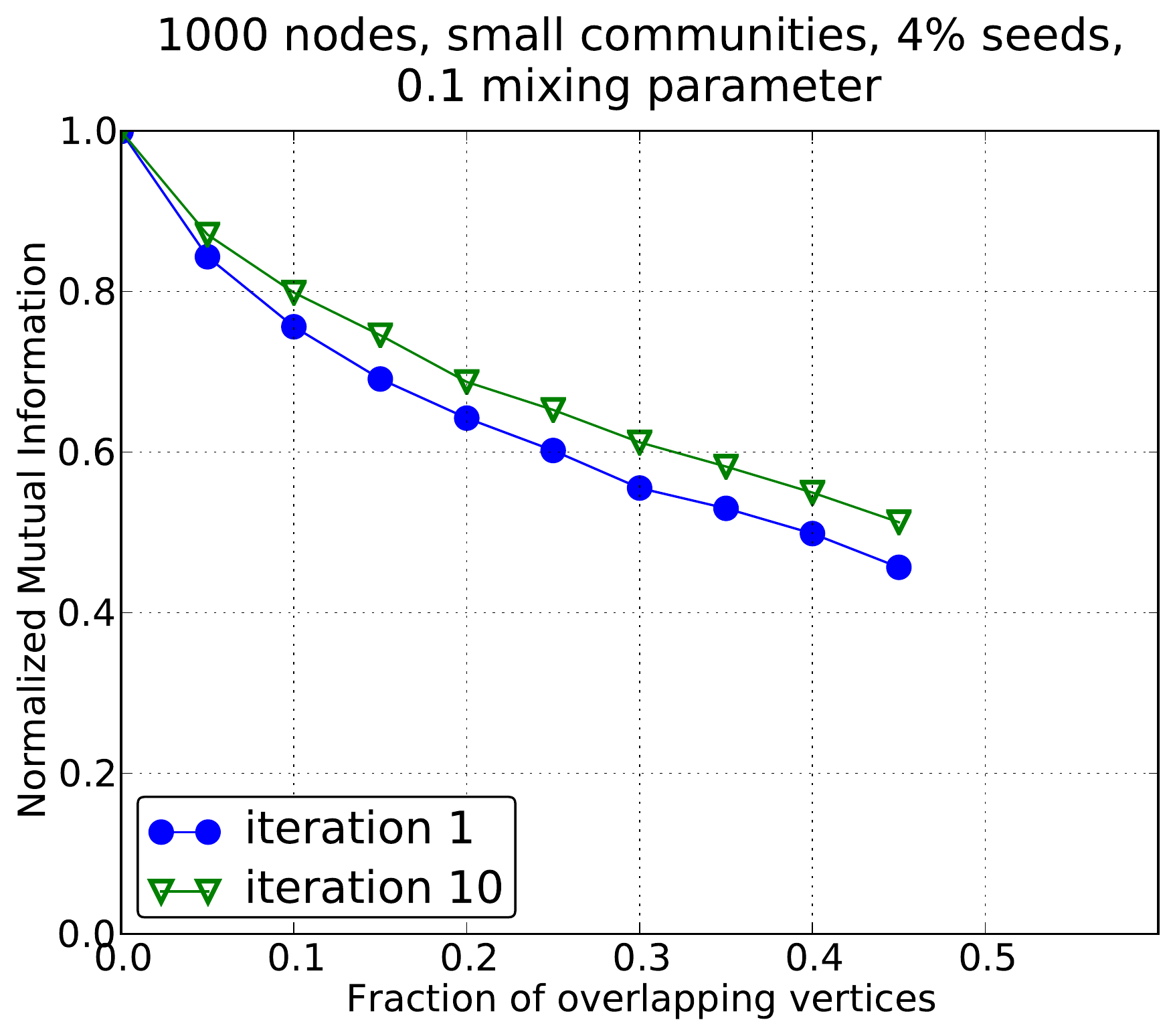}
    \end{subfigure}%
    \begin{subfigure}{0.5\textwidth}
    \centering
    \includegraphics[width=\plotwidth]{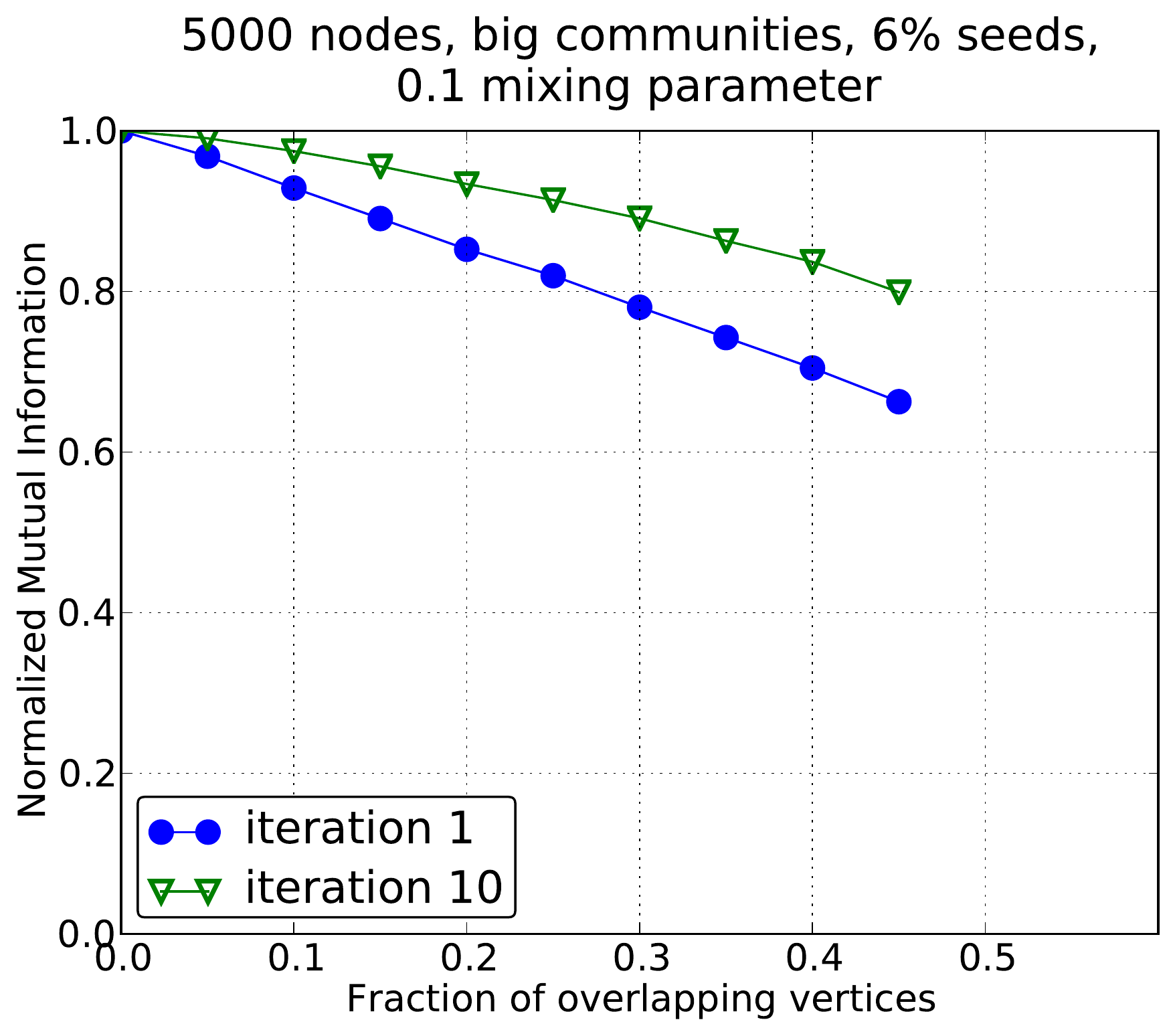}
    \end{subfigure}
    \caption{Comparison between the iterative and non-iterative method for overlapping communities.}\label{fig:compare_iter_overlap}
\end{figure}
\begin{figure}[h!]
    \centering
    \begin{subfigure}{0.5\textwidth}
    \centering
    \includegraphics[width=\cfinderwidth]{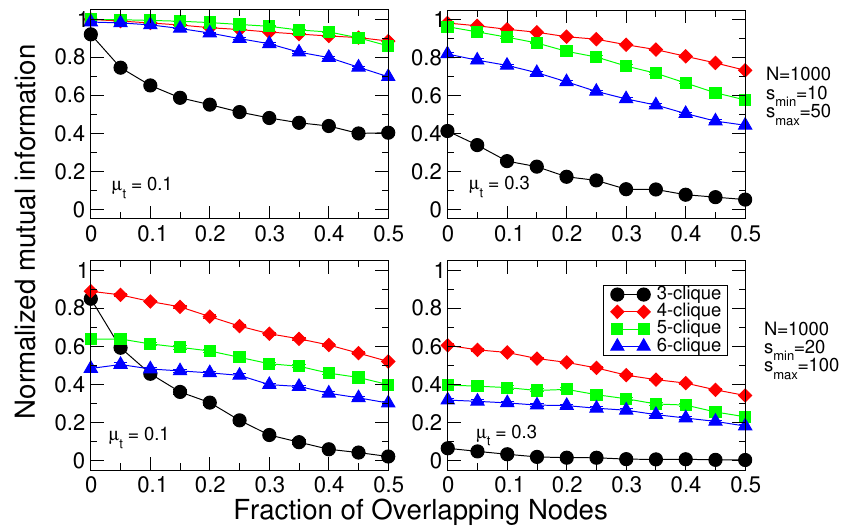}
    \end{subfigure}%
    \begin{subfigure}{0.5\textwidth}
    \centering
    \includegraphics[width=\cfinderwidth]{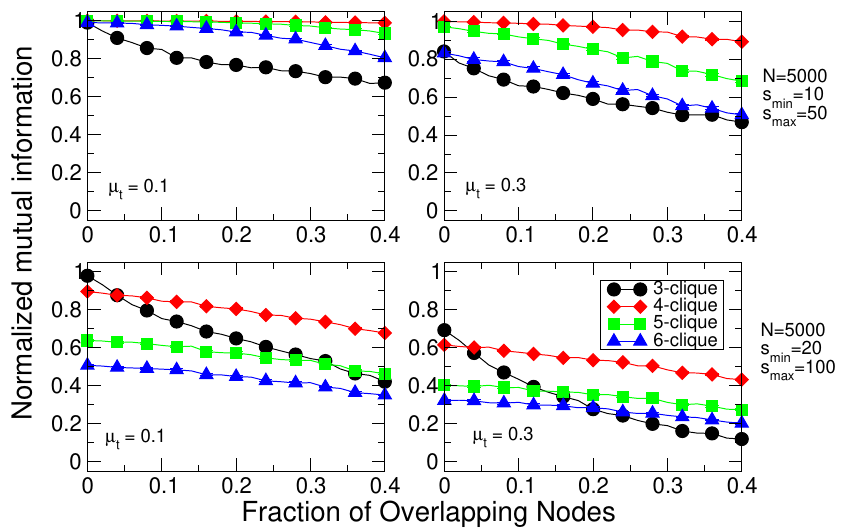}
    \end{subfigure}%
    \caption{
        Plots for CFinder on the LFR benchmark on graphs with 1000 and 5000 nodes 
		with overlapping communities. Reproduced from~\cite{LF09}.
    }\label{fig:CFinder_overlapping}
\end{figure}


%

\section{Concluding Remarks} \label{sec:conclusions}
Our algorithm seems to work very well with around 6$\%$ seed nodes for the 
non-overlapping case and around 8$\%$ seed nodes for the overlapping case. 
For the non-overlapping case, we can work with a mixing factor of up 
to~$0.5$, whereas in the overlapping case a mixing factor of $0.3$ 
and with the overlapping fraction of around $20\%$. This of course 
suggests that our algorithm has a higher tolerance while detecting 
non-overlapping communities and needs either a ``well-structured''
network or a high seed node percentage for overlapping communities. 
None of this is really surprising. What is surprising is that such 
a simple algorithm manages to do so well at all.

An obvious question is whether it is possible to avoid the semi-supervised 
step completely, that is, avoid having the user to specify seed nodes 
for every community. One possibility is to initially use a clustering 
algorithm to obtain a first approximation of the communities in the network. 
The next step would be to pick seed nodes from among the communities thus found 
(without user intervention) and use our algorithm to obtain a refinement of the 
community structure.

A possible extension of our algorithm would be to allow the user to interactively 
specify the seed nodes. The user initially supplies a set of seed nodes 
and allows the algorithm to find communities. The user then checks the 
quality of the output and, if dissatisfied with the results, can prompt the algorithm 
to correctly classify some more nodes that it had incorrectly classified in the current round. 
In effect, at the end of each round, the user supplies an additional set of seed nodes until the 
communities found out by the algorithm are accurate enough for the user. Such a tool 
might be useful for visualization.

We wish to point out that while the running time of our algorithm is 
$O(k \cdot m \cdot \log n)$, we do not know of any commercial solvers 
for SDD systems that run in $O(m \cdot \log n)$ time. Since we use the Cholesky 
factorization method from the \CPP\ Eigen Library, it is unlikely that our 
implementation would be able to handle very large networks. Recall that 
in Cholesky factorization, the matrix of coefficients $\mat{A}$ is decomposed 
as $\mat{L} \mat{D} \trans{\mat{L}}$, where $\mat{L}$ is lower triangular 
and $\mat{D}$ is diagonal, all of which takes $n^3/3$ operations making it 
prohibitively expensive for large networks (see, for instance~\cite{GvL13}). 
This is not a serious disadvantage since we expect that in the near future 
we would have commercial SDD solvers implementing the Speilman-Teng algorithm. 
It would then be interesting to see the size range of real networks our algorithm 
can handle.

\def\redefineme{
    \def\LNCS{LNCS}%
    \def\ICALP##1{Proc. of ##1 ICALP}%
    \def\FOCS##1{Proc. of ##1 FOCS}%
    \def\COCOON##1{Proc. of ##1 COCOON}%
    \def\SODA##1{Proc. of ##1 SODA}%
    \def\SWAT##1{Proc. of ##1 SWAT}%
    \def\IWPEC##1{Proc. of ##1 IWPEC}%
    \def\IWOCA##1{Proc. of ##1 IWOCA}%
    \def\ISAAC##1{Proc. of ##1 ISAAC}%
    \def\STACS##1{Proc. of ##1 STACS}
    \def\ESA##1{Proc. of ##1 ESA}%
    \def\WG##1{Proc. of ##1 WG}%
    \def\LIPIcs##1{LIPIcs}%
    \def\LIPIcs{LIPIcs}%
    \def\LICS##1{Proc. of ##1 LICS}%
}

\bibliographystyle{abbrv}
\bibliography{social_net}

\end{document}